\documentclass[12pt]{article}
\usepackage[english]{babel}
\usepackage[cp1251]{inputenc}
\usepackage{doc}
\usepackage{amsmath}
\usepackage{amsfonts}   
\usepackage{cite} 
\usepackage{fullpage}                           
\usepackage[lmargin=0.6in,rmargin=0.6in,tmargin=0.6in,headsep=.2in]{geometry}
\usepackage{graphicx}
\usepackage{forloop}
\usepackage{etoolbox}
\usepackage{calc}
\usepackage{wrapfig,lipsum,booktabs} 
\usepackage{xtab} 
\usepackage[dvipsnames]{xcolor}
\usepackage[normalem]{ulem}
\usepackage{sectsty}
\usepackage{amsmath,amsfonts,amssymb,amsthm,epsfig,epstopdf,url,array}
\usepackage[retainorgcmds]{IEEEtrantools}
\usepackage{makeidx,epsfig,lscape}
\usepackage{xcolor,pict2e}
\usepackage{upgreek}
 
\makeatother
\usepackage{makeidx}
\makeindex
\makeatletter
\renewcommand{\@biblabel}[1]{#1.}

\newtheorem{thm}{Theorem}    
\newtheorem{remark}{Remark}   
\newtheorem{lemma}{Lemma}   
\newtheorem{cor}{Corollary}  

\newtheorem{defin}{Definition}  

\begin{document}

\centerline{ \bf\Large{Description of Incomplete Financial Markets}}
\centerline{ \bf\Large{ for the  Discrete Time  Evolution of Risk Assets.}}

\vskip 21mm
{\bf \centerline {\Large  N.S. Gonchar } }
\vskip 5mm
\centerline{\bf {Bogolyubov Institute for Theoretical Physics of NAS of Ukraine.}}
\vskip 2mm

\begin{abstract}
In the paper, the martingales and super-martingales relative to a regular set of  measures are systematically studied.
The notion of  local  regular super-martingale relative to a  set of equivalent  measures is introduced and  the necessary and sufficient conditions of the local regularity of it  in the discrete case are founded.
The  regular set of  measures play fundamental role  for the description of incomplete markets. In the partial case, the description of the regular set of measures is presented. 
The notion of completeness of the regular set of measures have the  important significance for the simplification of the proof of the optional decomposition for super-martingales. 
Using this notion,  the important  inequalities for some random values are obtained. 
These inequalities give the simple proof of the optional decomposition   of the majorized super-martingales.
The description of all local regular super-martingales relative to the regular  set of  measures  is presented. 
It is  proved that every majorized  super-martingale   relative to the  complete  set of  measures  is a local regular one. 
In the case, as evolution of a risk asset is given by the discrete 
geometric Brownian motion, the financial market is incomplete  and a new formula  for the fair price of super-hedge is founded.
\end{abstract}

\centerline{{\bf Keywords:}  Random process; Regular set of measures; } 
\centerline{ Optional Doob decomposition;Local regular super-martingale; martingale; }
\centerline{Discrete geometric Brownian motion.}
\centerline{{\bf  2010 MSC}  60G07, 60G42.
\footnote{This work was supported  in part by The National Academy of Sciences of Ukraine (project No. 0118U003196).}}

\section{Introduction}

This paper is a continuation of the paper  \cite{GoncharNick}. In it, a new method of investigation of martingales and super-martingales relative to the regular set of  measures is developed.
A notion  of the local regular super-martingale relative to the regular  set of measures  is introduced and the necessary and sufficient conditions are found under that the above defined super-martingale is a local regular one. The last fact allowed us to describe the local regular super-martingales. 
On a measurable space, a notion of the set of equivalent measures consistent with the filtration is introduced.  Such a set of measures guarantee the  existence of  the sufficient set of nonnegative super-martingales. The next important fact is  the existence of a martingale on such a measurable space.  Further, we introduce the important notion of the regular set of measures. In  partial cases, we describe completely the set of regular measures.
An important notion of the completeness of the regular set  of  measures is introduced.
To prove that the regular  set of measures for the local regular martingale is a complete one we describe the set of equivalent measures to a given measure, which satisfy the condition: expectation of  a given random value  relative to  every measure   from this  set of  measures  equals zero.  
The representation for every measure of this  set of measures and a notion of  the  exhaustive  decomposition for the $\sigma$-algebra gives us the possibility  to prove the statement that the set of equivalent  martingale measures for the regular martingale  is a complete one.
 This notion is very important, since it permits us to find some important inequalities for a certain class of random variables. These inequalities simplify the proof of the optional decomposition for the class of  majorized super-martingales.

 The notion of the completeness of the regular  set of  measures permits us to give a new proof of the optional decomposition for a nonnegative super-martingale. This proof does not use the no-arbitrage arguments and  the measurable choice \cite{Kramkov}, \cite{FolmerKramkov1}, \cite{FolmerKabanov1},  \cite{FolmerKabanov}.

First, the optional decomposition for   diffusion processes super-martingale was opened by  by  El Karoui N. and  Quenez M. C. \cite{KarouiQuenez}. After that, Kramkov D. O. and Follmer H. \cite{Kramkov}, \cite{FolmerKramkov1} proved the optional decomposition for the nonnegative bounded super-martingales.  Folmer H. and Kabanov Yu. M.  \cite{FolmerKabanov1},  \cite{FolmerKabanov}  proved analogous result for an arbitrary super-martingale. Recently, Bouchard B. and Nutz M. \cite{Bouchard1} considered a class of discrete models and proved the necessary and sufficient conditions for the validity of the optional decomposition. 

The optional decomposition for super-martingales plays the fundamental role for the risk assessment in incomplete markets  \cite{Kramkov}, \cite{FolmerKramkov1},  \cite{KarouiQuenez},\cite{Gonchar2},  \cite{Gonchar555},  \cite{Gonchar557}, \cite{Honchar100}. 
Considered in the paper problem is a generalization of the corresponding one  that  appeared in mathematical finance  about the optional decomposition for a super-martingale and which is related with the construction of the super-hedge strategy in incomplete financial markets.

At last, we consider an application of the results obtained to find the new formula for the fair price of super-hedge in the case, as the risk asset evolves by the discrete geometric Brownian motion.

\section{Local regular super-martingales relative to  a  set of equivalent measures.}

We assume that on a measurable space $\{\Omega,\mathcal{F}\}$ a filtration    ${\mathcal{F}_{m}\subset\mathcal{F}_{m+1}}\subset\mathcal{F}, \ m=\overline{0, \infty},$ and a set  of equivalent  measures $ M$ on $\mathcal{F}$ are given. Further, we assume that ${\cal F}_0=\{\emptyset, \Omega \}$ and the $\sigma$-algebra  ${\cal F}=\sigma(\bigvee\limits_{n=1}^\infty {\cal F}_n)$ is a minimal $\sigma$-algebra generated by the algebra $\bigvee\limits_{n=1}^\infty {\cal F}_n.$
A random process $\psi={\{\psi_{m}\}_{m=0}^{\infty}}$ is said to be adapted one relative to the filtration $\{{\cal F}_m\}_{m=0}^{\infty},$ if $\psi_{m}$ is a ${\cal F}_m$ measurable random value, $m=\overline{0,\infty}.$
\begin{defin}
An adapted random process  $f={\{f_{m}\}_{m=0}^{\infty}}$ is said to be   a super-martingale relative to the filtration ${\cal F}_m,\ m=\overline{0,\infty},$ and the    family of equivalent  measures   $ M,$  if $E^P|f_m|<\infty, \ m=\overline{1, \infty}, \ P \in M,$ and the inequalities 
\begin{eqnarray}\label{pk11} 
E^P\{f_m|{\cal F}_k\} \leq f_k, \quad 0 \leq k \leq m, \quad m=\overline{1, \infty}, \quad P \in M,
\end{eqnarray}
are valid.
\end{defin}
 Further,   for an adapted process  $f$ we use both the denotation $\{f_{m}, {\cal F}_m\}_{m=0}^{\infty} $ and the denotation  $\{f_{m}\}_{m=0}^{\infty}.$

\begin{defin} A  super-martingale $\{f_{m},\ {\cal F}_m\}_{m=0}^{\infty}$ relative to a  set of equivalent measures M  is   a  local  regular one, if  $ \sup\limits_{P \in M}E^P|f_m| < \infty, \ m=\overline{1, \infty},$  and  there exists  an adapted nonnegative increasing  random process  $\{g_{m},\ {\cal F}_m\}_{m=0}^{\infty}, \ g_0=0,$  \  $ \sup\limits_{P \in M}E^P|g_m| < \infty,\ m=\overline{1, \infty},$ such that  $\{f_m+g_m, \ {\cal F}_m\}_{m=0}^{\infty}$
is a martingale relative to every measure from $M.$
\end{defin}
The next elementary Theorem \ref{reww1} will be very useful later.

\begin{thm}\label{reww1} Let a super-martingale  $\{f_{m},\ {\cal F}_m\}_{m=0}^{\infty}, $ relative to a  set of equivalent measures M  be such that   $ \sup\limits_{P \in M}E^P|f_m| < \infty, \ m=\overline{1, \infty}.$  The necessary and sufficient condition for  it   to be a local regular one is the existence of  an adapted nonnegative random process  $\{\bar g^0_{m},\ {\cal F}_m\}_{m=0}^{\infty},$  \  $ \sup\limits_{P \in M}E^P|\bar g^0_m| < \infty, \ m=\overline{1, \infty},$ such that
\begin{eqnarray}\label{o1}
f_{m-1} - E^P\{f_m|{\cal F}_{m-1}\}= E^P\{\bar g_m^0|{\cal F}_{m-1}\}, \quad m=\overline{1, \infty}, \quad P \in M.
\end{eqnarray}
\end{thm} 

\begin{proof}  The necessity. If   $\{f_{m},\ {\cal F}_m\}_{m=0}^{\infty}$ is  a local  regular super-martingale,  then there exist a martingale  $\{\bar M_{m},\ {\cal F}_m\}_{m=0}^{\infty}$ and a non-decreasing nonnegative random process  $\{g_{m},\ {\cal F}_m\}_{m=0}^{\infty},$ $ \ g_0=0,$ such that
\begin{eqnarray}\label{reww3}
f_m = \bar M_m - g_m, \quad m=\overline{1, \infty}.
\end{eqnarray}
From here, we  obtain  the equalities
$$E^P\{f_{m-1} -f_m|{\cal F}_{m-1}\}=$$
\begin{eqnarray}\label{reww4}
=E^P\{ g_m - g_{m-1}|{\cal F}_{m-1}\}=E^P\{\bar g_m^0|{\cal F}_{m-1}\}, \quad m=\overline{1, \infty} , \quad P \in M,
\end{eqnarray}
where we introduced the denotation $\bar g_m^0=g_m - g_{m-1} \geq 0.$
It is evident that $E^P\bar g_m^0\leq \sup\limits_{P\in M} E^Pg_m+\sup\limits_{P\in M} E^Pg_{m-1}< \infty.$

The sufficiency. Suppose that there exists an adapted nonnegative random process $\bar g^0=\{\bar g_m^0\}_{m=0}^\infty, \ \bar g_0^0=0,$  $ E^P\bar g_m^0<\infty, \ m=\overline{1, \infty}, $ such that the equalities (\ref{o1}) hold.  Let us consider the  random process $\{\bar M_{m},\ {\cal F}_m\}_{m=0}^{\infty},$ where
\begin{eqnarray}\label{reww5}
\bar M_0=f_0, \quad \bar M_m=f_m+\sum\limits_{i=1}^m\bar g_m^0, \quad m=\overline{1, \infty}.
\end{eqnarray}
It is evident that $E^P|\bar M_m|< \infty$ and
\begin{eqnarray}\label{apm1}
E^P\{\bar M_{m-1} -  \bar M_m|{\cal F}_{m-1}\}=E^P\{f_{m-1} - f_m- \bar g_m^0|{\cal F}_{m-1}\}=0.
\end{eqnarray}
  Theorem  \ref{reww1} is proved.
\end{proof}

\begin{lemma}\label{l1} Any super-martingale  ${\{f_m, {\cal F}_m\}_{m=0}^{\infty}}$ relative to  a family of measures  $ M$ for which there hold equalities   $E^{P}f_{m}=f_{0}, \ m=\overline{1,\infty},$ \  ${ P\in M},$ is a martingale  with respect to this family of measures and the filtration   ${\cal F}_m,\ m=\overline{1,\infty}.$
\end{lemma}
\begin{proof} The proof of  Lemma \ref{l1} see \cite{Kallianpur}.\end{proof}

In the next Lemma, we present the formula for calculation of the conditional expectation relative to another measure from $M.$
\begin{lemma}\label{q1}
 On the measurable space  $\{ \Omega, {\cal F}\}$ with the filtration ${\cal F}_n$ on it, 
let $M$ be a set of equivalent measures  and let  $\xi$ be an integrable random value.  Then, the following formulas 
\begin{eqnarray}\label{n1}
E^{P_1}\{\xi|{\cal F}_n\}=E^{P_2}\left\{\xi \varphi_n^{P_1}|{\cal F}_n\right\}, \quad n=\overline{1, \infty},   
\end{eqnarray}
are valid,  where

\begin{eqnarray}\label{apm2}
 \varphi_n^{P_1}=\frac{dP_1}{dP_2}\left[E^{P_2}\left\{\frac{dP_1}{dP_2}|{\cal F}_n\right\}\right]^{-1}, \quad P_1, \ P_2 \in M.
\end{eqnarray}
\end{lemma}
\begin{proof} The proof of  Lemma \ref{q1}   is evident.\end{proof}

\section{Local regular super-martingales relative to a   set of equivalent measures consistent with the filtration.} 

\begin{defin}\label{mykvita1}
On a measurable space $\{\Omega, {\cal F}\}$  
with a filtration ${\cal F}_n$ on it, a set of equivalent measures  $M$ we call consistent  with the filtration ${\cal F}_n,$ if for every pair of measures $(Q_1, Q_2)  \in M^2$ the set of measures
\begin{eqnarray}\label{tatnick1}
R_s^k(A)=\int\limits_{A}\frac{E^{Q_1}\{\frac{dQ_2}{dQ_1}|{\cal F}_k\}}{E^{Q_1}\{\frac{dQ_2}{dQ_1}|{\cal F}_s\}}dQ_1,  \quad A \in {\cal F},
 \quad k\geq s \geq n, \quad n=\overline{0, \infty},
\end{eqnarray}
belongs to the set $M,$ where $M^2$ is  a direct product of the set $M$ by itself.
\end{defin}
\begin{lemma}\label{tinnick1}
On the measurable space $\{\Omega, {\cal F}\}$  with  the filtration ${\cal F}_n$ on it, the set 
of measures
\begin{eqnarray}\label{tinnick2} 
M=\{ Q, \ Q(A)=\int\limits_{A}\alpha(\omega)dP, \ A \in {\cal F}, \ Q(\Omega)=1\}
\end{eqnarray}
 is a consistent one with the filtration ${\cal F}_n,$ if  $P$ is a measure on  $\{\Omega, {\cal F}\}$ and a random value $\alpha(\omega)$ runs over all nonnegative  random values, satisfying the condition $P(\{\omega, \alpha(\omega)>0\})=1.$
\end{lemma}
\begin{proof} Suppose that $(Q_1, Q_2)$   belongs to $ M^2. $ Then, $\frac{dQ_2}{dQ_1}=\frac{\alpha_2(\omega)}{\alpha_1(\omega)}$ and $ P(\{\omega, \frac{dQ_2}{dQ_1}>0\})=1,$ since the equalities  $P( \{\omega, 0 <  \alpha_1(\omega)< \infty\})=1,$ $P( \{\omega, 0 <  \alpha_2(\omega)< \infty\})=1$ are true.
It is evident that

$$R_s^k(A)=\int\limits_{A}\frac{E^{Q_1}\{\frac{dQ_2}{dQ_1}|{\cal F}_k\}}{E^{Q_1}\{\frac{dQ_2}{dQ_1}|{\cal F}_s\}}dQ_1=$$
\begin{eqnarray}\label{tinnick3}
\int\limits_{A}\frac{E^{Q_1}\{\frac{dQ_2}{dQ_1}|{\cal F}_k\}}{E^{Q_1}\{\frac{dQ_2}{dQ_1}|{\cal F}_s\}}\alpha_1(\omega)dP
,  \quad A \in {\cal F},
 \quad k\geq s \geq n, \quad n=\overline{0, \infty}.
\end{eqnarray}
It is easy to see that
\begin{eqnarray}\label{tinnick4}
P(\{\omega, \frac{E^{Q_1}\{\frac{dQ_2}{dQ_1}|{\cal F}_k\}}{E^{Q_1}\{\frac{dQ_2}{dQ_1}|{\cal F}_s\}}\alpha_1(\omega)>0 \} )=1, \quad k \geq s,
\end{eqnarray}
since
\begin{eqnarray}\label{tinnick5}
P(\{ \omega,  E^{Q_1}\{\frac{dQ_2}{dQ_1}|{\cal F}_k\}>0 \})=1, \quad k \geq s,
\end{eqnarray}
\begin{eqnarray}\label{tinnick6}
P(\{ \omega, 0 < E^{Q_1}\{\frac{dQ_2}{dQ_1}|{\cal F}_s\}< \infty)=1,  \quad s \geq n, \quad n=\overline{0, \infty}.
\end{eqnarray}
The last equality follows from the equivalence of the measures $Q_1, Q_2$ and $P.$ Altogether, it means that the set of measures $ R_s^k, \ k \geq s \geq n, \ n=\overline{0, \infty},$
belongs to the set  $M.$ The same is true for the pair $(Q_2, Q_1) \in M^2.$
Lemma \ref{tinnick1} is proved.
\end{proof} 

\begin{thm}\label{tatnick2}
 On the measurable space $\{\Omega, {\cal F}\}$ with the filtration ${\cal F}_n$ on it, let the set of equivalent measures $M$  be  consistent  with the filtration  ${\cal F}_n.$ Then, for every  nonnegative random value $\xi$ such that
$\sup\limits_{P \in M} E^P\xi < \infty,$ the random process $\{f_n, {\cal F}_n\}_{n=0}^\infty$
 is a super-martingale relative to the set of measures $M,$ where $f_n=\mathrm{ess}\sup\limits_{P \in M}E^P\{\xi| {\cal F}_n\}, \ n=\overline{0, \infty}.$
\end{thm}
\begin{proof}
Let $Q \in M,$ then,  due to Lemma  \ref{q1},  for every $ P \in M$ 
\begin{eqnarray}\label{tatnick3}
E^P\{\xi| {\cal F}_n\}=E^Q\left\{\xi| \frac{\frac{dP}{dQ}}{E^Q\{\frac{dP}{dQ}|{\cal F}_n\}}|{\cal F}_n\right\}.
\end{eqnarray}
If to put instead of the measure $P$ the measure $R_s^k, \ k \geq s \geq n,$ for the pair of measures $(Q, P)$ we obtain
\begin{eqnarray}\label{tatnick4}
E^{R_s^k}\{\xi| {\cal F}_n\}=E^Q\left\{\xi| \frac{\frac{dR_s^k}{dQ}}{E^Q\{\frac{dR_s^k}{dQ}|{\cal F}_n\}}|{\cal F}_n\right\}=E^Q\left\{\xi \frac{E^Q\{\frac{dP}{dQ}|{\cal F}_k\} }{E^Q\{\frac{dP}{dQ}|{\cal F}_s\}}|{\cal F}_n\right\},
\end{eqnarray}
where we took into account  the equality
\begin{eqnarray}\label{tatnick5}
E^Q\left\{\frac{dR_s^k}{dQ}|{\cal F}_n\right\}=E^Q\left\{\frac{E^Q\{ \frac{dP}{dQ}|{\cal F}_k\}}{E^Q\{\frac{dP}{dQ}|{\cal F}_s\}}|{\cal F}_n\right\}=1, \quad k \geq s \geq n.
\end{eqnarray}
From the formula (\ref{tatnick4}), it follows the equality 
\begin{eqnarray}\label{tatnick6}
\mathrm{ess}\sup\limits_{P \in M}E^P\{\xi| {\cal F}_n\}=\mathrm{ess}\sup\limits_{T \in R_n}E^P\{\xi T| {\cal F}_n\},
\end{eqnarray}
where $R_n $ is a set of martingales $T=\{T_m\}_{m=0}^\infty$ relative to the measure $Q$ such that $T_m=1, m \leq n, T_m=\frac{E^Q\{\frac{dP}{dQ}|{\cal F}_m\}}{E^Q\{\frac{dP}{dQ}|{\cal F}_s\}}, \  m \geq s \geq n, \ P \in M.$
The definition of $\mathrm{ess}\sup$ for the uncountable set of random values see \cite{Chow}. It is evident that $T_n \subseteq T_{n-1}.$
Let us consider

$$ E^Q\{\mathrm{ess}\sup\limits_{P \in M}E^P\{\xi| {\cal F}_n\}| {\cal F}_{n-1}\}=
E^Q\{\mathrm{ess}\sup\limits_{T \in R_n}E^P\{\xi T| {\cal F}_n\}| {\cal F}_{n-1}\}=
$$
$$
E^Q\{\sup\limits_{i \geq 1}E^P\{\xi T_i| {\cal F}_n\}| {\cal F}_{n-1}\}=
E^Q\{\lim\limits_{k\to \infty} \max\limits_{1 \leq i \leq k}E^P\{\xi T_i| {\cal F}_n\}| {\cal F}_{n-1}\}=$$
$$\lim\limits_{k\to \infty}E^Q\{ \max\limits_{1 \leq i \leq k}E^P\{\xi T_i| {\cal F}_n\}| {\cal F}_{n-1}\}=\lim\limits_{k\to \infty}E^P\{\xi T_{\tau_k}| {\cal F}_{n-1}\}\leq $$
$$\mathrm{ess}\sup\limits_{T \in R_n}E^Q\{\xi T|{\cal F}_{n-1}\} \leq 
\mathrm{ess}\sup\limits_{T \in R_{n-1}}E^Q\{\xi T|{\cal F}_{n-1}\}=$$
\begin{eqnarray}\label{tatnick7}
\mathrm{ess}\sup\limits_{P \in M}E^Q\{\xi |{\cal F}_{n-1}\},
\end{eqnarray}
where
\begin{eqnarray}\label{apm1}
\tau_1=1,
\end{eqnarray}
 \begin{eqnarray}\label{apm2}
 \tau_i=\left\{\begin{array}{l l} \tau_{i-1},  &   E^{P}\{\xi T_{\tau_{i-1}}|{\cal F}_n\} > E^{P}\{\xi T_i |{\cal F}_n\}, \\ 
i, &  E^{P}\{ \xi T_{\tau_{i-1}}|{\cal F}_n\} \leq E^{P}\{\xi T_i|{\cal F}_n\},
\end{array} \right. \quad i=\overline{2,k}.
\end{eqnarray}
Lemma \ref{tatnick2} is proved.
\end{proof}

\begin{thm}\label{tatnick8}
On the measurable space  $\{\Omega, {\cal F}\}, \ {\cal F}=\sigma(\bigvee\limits_{i=1}^\infty{\cal F}_i),$ let $M$ be a  set of equivalent measures being  consistent with the filtration ${\cal F}_n.$  If there exists a  nonnegative random value $\xi \neq 1$ such that $E^P \xi=1, P \in M,$ then  $E^P \{\xi |{\cal F}_n\}, P \in M,$ is a local regular martingale.
\end{thm}
\begin{proof}
Due to Lemma \ref{tatnick2}, the  random process  $\{f_n, {\cal F}_n\}_{n=0}^\infty,$ where $f_n=$ $\mathrm{ess}\sup\limits_{P \in M}E^P\{\xi| {\cal F}_n\},$ $ \ n=\overline{0, \infty},$ is a  super-martingale relative to
 the set of  measures $M,$ that is, 
\begin{eqnarray}\label{tatnick9}
E^Q\{\mathrm{ess}\sup\limits_{P \in M}E^P\{\xi| {\cal F}_n\}|{\cal F}_{n-1}\}\leq
\mathrm{ess}\sup\limits_{P \in M}E^P\{\xi| {\cal F}_{n-1}\}, \quad Q \in M, \quad n=\overline{0,\infty}.
\end{eqnarray}
From  the inequality (\ref{tatnick9}), it follows the inequality
\begin{eqnarray}\label{tatnick10}
E^Q\mathrm{ess}\sup\limits_{P \in M}E^P\{\xi| {\cal F}_n\} \leq 1, \quad n=\overline{0,\infty}.
\end{eqnarray}
Since $E^Q\mathrm{ess}\sup\limits_{P \in M}E^P\{\xi| {\cal F}_n\}  \geq E^QE^Q\{\xi| {\cal F}_n\}=1,$ we have
\begin{eqnarray}\label{tatnick11}
E^Q\mathrm{ess}\sup\limits_{P \in M}E^P\{\xi| {\cal F}_n\}=1, \quad Q \in M, \quad n=\overline{0,\infty}.
\end{eqnarray}
The inequalities (\ref{tatnick9}) and the equalities (\ref{tatnick11}) give the equalities
\begin{eqnarray}\label{tatnick12}
E^Q\{\mathrm{ess}\sup\limits_{P \in M}E^P\{\xi| {\cal F}_n\}|{\cal F}_{n-1}\}=
\mathrm{ess}\sup\limits_{P \in M}E^P\{\xi| {\cal F}_{n-1}\}, \quad Q \in M, \quad n=\overline{1,\infty},
\end{eqnarray}
which are true with the probability 1.
The last  means that $\{f_n, {\cal F}_n\}_{n=0}^\infty$
 is a martingale relative to the set of measures $M,$ where $f_n=\mathrm{ess}\sup\limits_{P \in M}E^P\{\xi| {\cal F}_n\}, \ n=\overline{0, \infty}.$
With the  probability 1,   $\lim\limits_{n\to \infty}\mathrm{ess}\sup\limits_{P \in M}E^P\{\xi| {\cal F}_n\}=f_{\infty},$ where the random value  $f_{\infty}$ is ${\cal F}$ measurable one.
From the inequality (\ref{tatnick10})  and Fatou Lemma \cite{Gonchar7}, \cite{Chow}, we obtain
\begin{eqnarray}\label{tatnick13}
E^Pf_{\infty} \leq 1, \quad P \in M.
\end{eqnarray}
Prove that $f_{\infty}=\xi.$ Going to the limit in the inequality 
\begin{eqnarray}\label{tatnick14}
\mathrm{ess}\sup\limits_{P \in M}E^P\{\xi| {\cal F}_n\} \geq E^{P_1}\{\xi| {\cal F}_n\},
\end{eqnarray}
as $n \to \infty,$  we obtain the inequality
\begin{eqnarray}\label{tatnick15}
f_{\infty} \geq \xi.
\end{eqnarray}
From the inequality (\ref{tatnick13}) and the inequality (\ref{tatnick15}), we obtain
the inequalities $1\geq E^Pf_{\infty} \geq E^P\xi=1.$ Or, $E^Pf_{\infty}=1.$
The equalities $E^Pf_{\infty}=1, E^P\xi=1$ and the inequality (\ref{tatnick15}) give
the equality $f_{\infty}= \xi$ with the probability 1. Lemma \ref{tatnick8} is proved.
\end{proof}

\begin{lemma}\label{vitanick1}
On the measurable space   $\{\Omega, {\cal F}\}$ with the  filtration ${\cal F}_n$ on it,
let there exist $k$ equivalent measures $P_1, \ldots, P_k, k>1,$  and  a nonnegative random value $\xi_0 \neq 1 $ be  such that  
\begin{eqnarray}\label{vitanick2}
E^{P_i}\{\xi_0|{\cal F}_n\}=E^{P_1}\{\xi_0|{\cal F}_n\}, \quad E^{P_i}\xi_0=1, \quad i=\overline{2,k},\quad n=\overline{0, \infty}.
\end{eqnarray}
Then, there exists  the  set of equivalent measures $M$ consistent with the filtration  ${\cal F}_n,$ satisfying the condition $E^P\xi_0=1, P \in M.$
\end{lemma}
\begin{proof}
Let us consider the set of equivalent measures $M,$  satisfying the condition
\begin{eqnarray}\label{vitanick3}
E^{P}\{\xi_0|{\cal F}_n\}=E^{P_1}\{\xi_0|{\cal F}_n\}, \quad n=\overline{0, \infty}, \quad P \in M.
\end{eqnarray}
Such a set of measures is a nonempty one. Suppose that $Q_1, Q_2 \in M,$ then 
\begin{eqnarray}\label{vitanick4}
E^{Q_1}\{\xi_0|{\cal F}_n\}=E^{Q_2}\{\xi_0|{\cal F}_n\}, \quad n=\overline{0, \infty}. 
\end{eqnarray}
Let us prove that the formula 
\begin{eqnarray}\label{vitanick5}
E^{Q_1}\left\{\xi_0 \frac{E^{Q_1}\{\frac{dQ_2}{dQ_1}|{\cal F}_k\} }{E^{Q_1}\{\frac{dQ_2}{dQ_1}|{\cal F}_s\}}|{\cal F}_n\right\}=E^{Q_1}\{\xi_0|{\cal F}_n\}, \quad n \leq s \leq k, \quad n=\overline{0, \infty},
\end{eqnarray}
is valid.
Let $ s \geq n. $ Then, from the equalities  (\ref{vitanick4}), we have
$$ E^{Q_1}\{ E^{Q_2}\{\xi_0|{\cal F}_s\}|{\cal F}_n\}=E^{Q_1}\{\xi_0|{\cal F}_n\}.$$
Let $ k \geq s.$ Then,
$$ E^{Q_1}\{ E^{Q_2}\{\xi_0|{\cal F}_s\}|{\cal F}_n\}=E^{Q_1}\{ E^{Q_2}\{E^{Q_2}\{\xi_0|{\cal F}_k\} |{\cal F}_s\}|{\cal F}_n\}=$$
$$E^{Q_1}\{ E^{Q_1}\{E^{Q_2}\{\xi_0|{\cal F}_k\} \frac{ \frac{dQ_2}{dQ_1}}{E^{Q_1}\{\frac{dQ_2}{dQ_1}|{\cal F}_s\} }|{\cal F}_s\}|{\cal F}_n\}=$$
$$ E^{Q_1}\{E^{Q_2}\{\xi_0|{\cal F}_k\} \frac{ \frac{dQ_2}{dQ_1}}{E^{Q_1}\{\frac{dQ_2}{dQ_1}|{\cal F}_s\}}|{\cal F}_n\}=$$
$$ E^{Q_1}\{E^{Q_1}\{\xi_0|{\cal F}_k\} \frac{ \frac{dQ_2}{dQ_1}}{E^{Q_1}\{\frac{dQ_2}{dQ_1}|{\cal F}_s\}}|{\cal F}_n\}=$$

$$ E^{Q_1}\{E^{Q_1}\{\xi_0|{\cal F}_k\} \frac{E^{Q_1}\{ \frac{dQ_2}{dQ_1}|{\cal F}_k\}}{E^{Q_1}\{\frac{dQ_2}{dQ_1}|{\cal F}_s\}}|{\cal F}_n\}=$$

$$ E^{Q_1}\{\xi_0|\frac{E^{Q_1}\{ \frac{dQ_2}{dQ_1}|{\cal F}_k\}}{E^{Q_1}\{\frac{dQ_2}{dQ_1}|{\cal F}_s\}}|{\cal F}_n\}.$$
This proves the formula (\ref{vitanick5}).
To finish the proof of Lemma \ref{vitanick1}, it needs to prove that the set of  measures
\begin{eqnarray}\label{vitanick6}
R_s^k(A)=\int\limits_{A}\frac{E^{Q_1}\{\frac{dQ_2}{dQ_1}|{\cal F}_k\}}{E^{Q_1}\{\frac{dQ_2}{dQ_1}|{\cal F}_s\}}dQ_1,  \quad A \in {\cal F},
 \quad k\geq s \geq n, \quad n=\overline{0, \infty},
\end{eqnarray}
belongs to the set $M.$
Really,

$$ E^{R_s^k}\{\xi_0| {\cal F}_n\}=E^{Q_1}\left\{\xi_0| \frac{\frac{dR_s^k}{dQ_1}}{E^{Q_1}\{\frac{dR_s^k}{dQ_1}|{\cal F}_n\}}|{\cal F}_n\right\}=$$
\begin{eqnarray}\label{vitanick7}
E^{Q_1}\left\{\xi_0 \frac{E^{Q_1}\{\frac{dQ_2}{dQ_1}|{\cal F}_k\} }{E^{Q_1}\{\frac{dQ_2}{dQ_1}|{\cal F}_s\}}|{\cal F}_n\right\}=
E^{Q_1}\{\xi_0|{\cal F}_n\},  
\end{eqnarray}
where we took into account  the equality
\begin{eqnarray}\label{vitanick8}
E^{Q_1}\left\{\frac{dR_s^k}{dQ_1}|{\cal F}_n\right\}=E^{Q_1}\left\{\frac{E^{Q_1}\{ \frac{dQ_2}{dQ_1}|{\cal F}_k\}}{E^{Q_1}\{\frac{dQ_2}{dQ_1}|{\cal F}_s\}}|{\cal F}_n\right\}=1, \quad k \geq s \geq n.
\end{eqnarray}
From this, it follows that the set of measures $R_s^k \in M.$ This proves the consistence with the filtration of the set of measures $M.$ Lemma \ref{vitanick1} is proved.
\end{proof}

On a probability space   $\{\Omega, {\cal F}, P\},$ 
let $\xi$ be a random value, satisfying the conditions 
\begin{eqnarray}\label{vitanick10}
0 < P(\{\omega, \xi >0\})<1, \quad 0 < P(\{\omega, \xi <0\}).
\end{eqnarray}
Denote $\Omega^+=\{\omega, \xi(\omega)>0\}, \ \Omega^-=\{\omega, \xi(\omega)\leq 0\}$ and let ${\cal F}^-,$ $  {\cal F}^+$ be the restrictions of the $\sigma$-algebra ${\cal F}$ on the sets  $\Omega^-$ and  $\Omega^+,$ correspondingly. Suppose that $P^-$ and $P^+$ are the contractions of the measure $P$ on the $\sigma$-algebras ${\cal F}^-,$ $  {\cal F}^+,$ correspondingly. Consider the measurable space with measure  $\{  \Omega^-\times  \Omega^+, {\cal F}^-\times{\cal F}^+, \mu\},$ which is a direct product of the measurable spaces with measures  $\{\Omega^-, {\cal F}^-, P^-\}$ and $\{\Omega^+, {\cal F}^+, P^+\},$ where $\mu=P^-\times P^+.$ 
Introduce the  denotations
\begin{eqnarray}\label{100vitanick1}
\xi^+(\omega)=\left\{\begin{array}{l l} \xi(\omega),  &  \omega \in \{\xi(\omega)>0 \}, \\ 
0, & \omega \in \{\xi(\omega)\leq 0 \},
\end{array} \right.
\end{eqnarray}

\begin{eqnarray}\label{100vitanick2}
\xi^-(\omega)=\left\{\begin{array}{l l} -\xi(\omega),  &  \omega \in \{\xi(\omega)\leq 0 \},\\ 
0, & \omega \in \{\xi(\omega) > 0 \}.
\end{array} \right.
\end{eqnarray}
Then, $\xi(\omega)=\xi^+(\omega) -\xi^-(\omega).$ 

 On the measurable space $\{  \Omega^-\times  \Omega^+, {\cal F}^-\times{\cal F}^+, P^-\times P^+\},$  we assume that  the set of nonnegative measurable functions $\alpha(\omega_1,\omega_2),$ satisfying the conditions
\begin{eqnarray}\label{vitanick12}
\mu(\{(\omega_1,\omega_2) \in \Omega^- \times \Omega^+, \    \alpha(\omega_1, \omega_2)>0\})=P(\Omega^+)P(\Omega^-), 
\end{eqnarray}
\begin{eqnarray}\label{vitanick13}  \int\limits_{\Omega^-}\int\limits_{\Omega^+}\alpha(\omega_1, \omega_2)\frac{\xi^-(\omega_1) \xi^+(\omega_2)}{\xi^-(\omega_1)+\xi^+(\omega_2)}d\mu(\omega_1,\omega_2)<\infty,
\end{eqnarray}

\begin{eqnarray}\label{vitanick14}
\int\limits_{\Omega^-}\int\limits_{\Omega^+}\alpha(\omega_1, \omega_2)d\mu(\omega_1,\omega_2)=1,
\end{eqnarray}
is a nonempty set. Such assumptions are true for the nonempty set of  bounded random values $\alpha(\omega_1,\omega_2),$ for example, if the random value $\xi$ is an integrable one relative to the measure $P.$ 

\begin{lemma}\label{vitanick9}
On the probability space  $\{\Omega, {\cal F}, P\},$ let a random value $\xi$  satisfy the conditions (\ref{vitanick10}) and let
a measure $Q$  be equivalent to the measure $P$ and such that $E^Q\xi=0.$
Then, for the measure $Q$ the following representation 
$$Q(A)=\int\limits_{\Omega^-}\int\limits_{\Omega^+}\chi_{A}(\omega_1)\alpha(\omega_1, \omega_2)\frac{\xi^+(\omega_2)}{\xi^-(\omega_1)+\xi^+(\omega_2)}d\mu(\omega_1,\omega_2)+$$
\begin{eqnarray}\label{vitanick11} \int\limits_{\Omega^-}\int\limits_{\Omega^+}\chi_{A}(\omega_2)\alpha(\omega_1, \omega_2)\frac{\xi^-(\omega_1)}{\xi^-(\omega_1)+\xi^+(\omega_2)}d\mu(\omega_1,\omega_2), \quad A \in {\cal F},
\end{eqnarray}
is valid for those  random value  $\alpha(\omega_1, \omega_2)$  
that satisfy the conditions  (\ref{vitanick12}) - (\ref{vitanick14}).

Every measure $Q,$ given by the formula (\ref{vitanick11}), with the random value $\alpha(\omega_1, \omega_2),$ satisfying the conditions   (\ref{vitanick12}) -   (\ref{vitanick14}) is equivalent to the measure $P$ and is such that $E^Q\xi=0.$ For the measure $Q,$ the canonical representation 
$$Q(A)=\int\limits_{\Omega^-}\int\limits_{\Omega^+}\chi_{A}(\omega_1)\alpha_1(\omega_1, \omega_2)\frac{\xi^+(\omega_2)}{\xi^-(\omega_1)+\xi^+(\omega_2)}d\mu(\omega_1,\omega_2)+$$
\begin{eqnarray}\label{100vitanick3} \int\limits_{\Omega^-}\int\limits_{\Omega^+}\chi_{A}(\omega_2)\alpha_1(\omega_1, \omega_2)\frac{\xi^-(\omega_1)}{\xi^-(\omega_1)+\xi^+(\omega_2)}d\mu(\omega_1,\omega_2), \quad A \in {\cal F},
\end{eqnarray}
is valid, where
\begin{eqnarray}\label{100vitanick4}
\alpha_1(\omega_1, \omega_2)=\frac{\psi_1(\omega_1)\psi_2(\omega_2)[\xi^-(\omega_1)+\xi^+(\omega_2)]}{d}, \quad (\omega_1, \omega_2) \in \Omega^-\times \Omega^+,
\end{eqnarray}
\begin{eqnarray}\label{100vitanick5}
\psi_1(\omega_1)=\int\limits_{\Omega^+}\alpha(\omega_1, \omega_2)\frac{\xi^+(\omega_2)}{\xi^-(\omega_1)+\xi^+(\omega_2)}dP(\omega_2), \quad \omega_1 \in \Omega^-,
\end{eqnarray}
\begin{eqnarray}\label{100vitanick6}
\psi_2(\omega_2)=\int\limits_{\Omega^-}\alpha(\omega_1, \omega_2)\frac{\xi^-(\omega_1)}{\xi^-(\omega_1)+\xi^+(\omega_2)}dP(\omega_1), \quad \omega_2 \in \Omega^+,
\end{eqnarray}
\begin{eqnarray}\label{100vitanick7}
d=\int\limits_{\Omega^-}\xi^-(\omega_1)\psi_1(\omega_1)dP(\omega_1)=
\int\limits_{\Omega^+}\xi^+(\omega_2)\psi_2(\omega_2)dP(\omega_2).
\end{eqnarray}
\end{lemma}
\begin{proof}
From  the Lemma \ref{vitanick9} conditions,
\begin{eqnarray}\label{vitanick15}
Q(A)=\int\limits_{A}\psi(\omega)dP, \quad P(\{\omega, \psi(\omega)>0\})=1,
\end{eqnarray}
\begin{eqnarray}\label{vitanick16}
\int\limits_{\Omega}\psi(\omega)\xi(\omega)dP(\omega)=0.
\end{eqnarray}
The  condition (\ref{vitanick16}) means
\begin{eqnarray}\label{vitanick17}
\int\limits_{\Omega^+}\psi_2(\omega_2)\xi^+(\omega_2)dP(\omega_2)=
\int\limits_{\Omega^-}\psi_1(\omega_1)\xi^-(\omega_1)dP(\omega_1)=d>0,
\end{eqnarray}
where
\begin{eqnarray}\label{100vitanick8}
\psi_1(\omega)=\left\{\begin{array}{l l} \psi(\omega),  &  \omega \in \Omega^-, \\ 
0, & \omega \in \Omega^+,
\end{array} \right.
\end{eqnarray}
\begin{eqnarray}\label{100vitanick9}
\psi_2(\omega)=\left\{\begin{array}{l l} \psi(\omega),  &  \omega \in \Omega^+, \\ 
0, & \omega \in \Omega^-.
\end{array} \right.
\end{eqnarray}
Let us put 
\begin{eqnarray}\label{vitanick18}
\alpha(\omega_1, \omega_2)=\frac{\psi_1(\omega_1)\psi_2(\omega_2)[\xi^-(\omega_1+\xi^+(\omega_2)]}{d}, \quad   (\omega_1,\omega_2) \in \Omega^- \times \Omega^+.
\end{eqnarray}
Then, for such $\alpha(\omega_1, \omega_2) $ the equality (\ref{vitanick12}) is true.
Moreover,
\begin{eqnarray}\label{vitanick19}  \int\limits_{\Omega^-}\int\limits_{\Omega^+}\alpha(\omega_1, \omega_2)\frac{\xi^-(\omega_1)\xi^+(\omega_2)}{\xi^-(\omega_1)+\xi^+(\omega_2)}d\mu(\omega_1, \omega_2)=d^2<\infty,
\end{eqnarray}
$$\int\limits_{\Omega^-}\int\limits_{\Omega^+}\alpha(\omega_1, \omega_2)d\mu(\omega_1, \omega_2)=$$
\begin{eqnarray}\label{vitanick20}
\int\limits_{\Omega^-}\psi_1(\omega_1)dP(\omega_1)+
\int\limits_{\Omega^+}\psi_2(\omega_2)dP(\omega_2)=1,
\end{eqnarray}
$$E^Q\xi=\int\limits_{\Omega^-}\int\limits_{\Omega^+}\alpha(\omega_1, \omega_2)\xi(\omega_1)\frac{\xi^+(\omega_2)}{\xi^-(\omega_1)+\xi^+(\omega_2)}d\mu(\omega_1, \omega_2)+$$
\begin{eqnarray}\label{vitanick21} \int\limits_{\Omega^-}\int\limits_{\Omega^+}\alpha(\omega_1, \omega_2)\xi(\omega_2)\frac{\xi^-(\omega_1)}{\xi^-(\omega_1)+\xi^+(\omega_2)}d\mu(\omega_1, \omega_2)=0,
\end{eqnarray}
since $\xi(\omega_1)=-\xi^-(\omega_1), \ \omega_1 \in  \Omega^-, \ \xi(\omega_2)=\xi^+(\omega_2), \ \omega_2 \in \Omega^+.$ 

Let us prove  the last statement of  Lemma \ref{vitanick9}.
Suppose that the representation (\ref{vitanick11}) for the measure $Q,$ satisfying the conditions (\ref{vitanick12}) - (\ref{vitanick14}),  is valid.
 Taking into account  the denotations  (\ref{100vitanick5}) - (\ref{100vitanick7}),  we obtain
\begin{eqnarray}\label{100vitanick9}
Q(A)=\int\limits_{\Omega^-}\chi_{A}(\omega_1)\psi_1(\omega_1)dP(\omega_1)+\int\limits_{\Omega^+}\chi_{A}(\omega_2)\psi_2(\omega_2)dP(\omega_2),
\end{eqnarray}
$$ 0=E^Q\xi=\int\limits_{\Omega^-}\xi(\omega_1)\psi_1(\omega_1)dP(\omega_1)+
\int\limits_{\Omega^+}\xi(\omega_2)\psi_2(\omega_2)dP(\omega_2)=$$
\begin{eqnarray}\label{100vitanick10}
-\int\limits_{\Omega^-}\xi^-(\omega_1)\psi_1(\omega_1)dP(\omega_1)+
\int\limits_{\Omega^+}\xi^+(\omega_2)\psi_2(\omega_2)dP(\omega_2).
\end{eqnarray}
If to introduce the denotation
\begin{eqnarray}\label{100vitanick11}
\psi(\omega)=\left\{\begin{array}{l l} \psi_1(\omega),  &  \omega \in \Omega^-, \\ 
\psi_2(\omega), & \omega \in \Omega^+,
\end{array} \right.
\end{eqnarray}
then we obtain the representation
\begin{eqnarray}\label{100vitanick12}
Q(A)= \int\limits_{A}\psi(\omega)dP(\omega),
\end{eqnarray}
where $P(\psi_1(\omega)>0)=P(\Omega^-), \  P(\psi_2(\omega)>0)=P(\Omega^+).$

The last formula proves the equivalence of the measures $Q$ and $P.$
At last, to prove the  canonical representation (\ref{100vitanick3}) it is sufficient to substitute the expression  (\ref{100vitanick4}) for $\alpha_1(\omega_1, \omega_2)$ into the expression (\ref{100vitanick3}) for $Q(A).$ We obtain the expression (\ref{100vitanick9}) for $Q(A).$ Then, if to substitute the expressions (\ref{100vitanick5}), (\ref{100vitanick6}) for   $\psi_1(\omega_1), \psi_2(\omega_2)$ into the expression (\ref{100vitanick9}) for $Q(A),$ we obtain that the canonical representation for $Q(A)$ is true.   This proves
 Lemma \ref{vitanick9}.
\end{proof}

Let  $\{\Omega, {\cal F}, P\}$ be a probability space and let  $G$ be a sub $\sigma$-algebra of the $\sigma$-algebra $ {\cal F}.$

\begin{lemma}\label{101vitanick13}
On the probability space $\{\Omega, {\cal F}, P\},$ let a  random value $\xi$ satisfy  the conditions (\ref{vitanick10}) and let it be an integrable one relative to the measure $P.$  A measure $Q,$ being  equivalent to the measure $P,$
satisfies the condition
\begin{eqnarray}\label{100vitanick13}
E^Q\{\xi|G\}=0
\end{eqnarray}
if and only if for every $B \in G$ such that $P(B)>0$ for the measure  $Q$ the representation 
$$Q(A)=\int\limits_{\Omega^{B,-}}\int\limits_{\Omega^{B,+}}\chi_{A}(\omega_1)\alpha_1(\omega_1, \omega_2)\frac{\zeta^{B,+}(\omega_2)}{\zeta^{B,-}(\omega_1)+\zeta^{B,+}(\omega_2)}d\mu(\omega_1,\omega_2)+$$
\begin{eqnarray}\label{100vitanick14} \int\limits_{\Omega^{B,-}}\int\limits_{\Omega^{B,+}}\chi_{A}(\omega_2)\alpha_1(\omega_1, \omega_2)\frac{\zeta^{B,-}(\omega_1)}{\zeta^{B,-}(\omega_1)+\zeta^{B,+}(\omega_2)}d\mu(\omega_1,\omega_2), \ A \in {\cal F},
\end{eqnarray}
is true and the equalities
\begin{eqnarray}\label{100vitanick15}
\alpha_1(\omega_1, \omega_2)=\frac{\psi_1(\omega_1)\psi_2(\omega_2)[\zeta^{B,-}(\omega_1+\zeta^{B,+}(\omega_2)]}{d^B}, 
\end{eqnarray}
$$   (\omega_1,\omega_2) \in \Omega^{B,-} \times \Omega^{B,+},$$
\begin{eqnarray}\label{100vitanick16}
d^B=\int\limits_{\Omega^{B,-}}\zeta^{B,-}(\omega_1)\psi_1(\omega_1)dP(\omega_1)=
\int\limits_{\Omega^{B,+}}\zeta^{B,+}(\omega_2)\psi_2(\omega_2)dP(\omega_2),
\end{eqnarray}
are valid, where
\begin{eqnarray}\label{100vitanick17}
\zeta^{B,+}(\omega)=\left\{\begin{array}{l l} \xi(\omega),  &  \omega \in B\cap\{\xi(\omega)>0\}, \\ 
0, & \omega \in (\Omega\setminus B)\cup \{ \xi(\omega) \leq 0\},
\end{array} \right.
\end{eqnarray}
\begin{eqnarray}\label{100vitanick17}
\zeta^{B,-}(\omega)=\left\{\begin{array}{l l} -\xi(\omega),  &  \omega \in B\cap\{\xi(\omega)\leq 0\}, \\ 
0, & \omega \in (\Omega\setminus B)\cup \{ \xi(\omega) > 0\},
\end{array} \right.
\end{eqnarray}
\begin{eqnarray}\label{100vitanick18}
\psi_1(\omega)=\left\{\begin{array}{l l} \psi(\omega),  &  \omega \in \Omega^{B,-}, \\ 
0, & \omega \in \Omega^{B,+},
\end{array} \right.
\end{eqnarray}
\begin{eqnarray}\label{100vitanick19}
\psi_2(\omega)=\left\{\begin{array}{l l} \psi(\omega),  &  \omega \in \Omega^{B,+}, \\ 
0, & \omega \in \Omega^{B,-},
\end{array} \right.
\end{eqnarray}
\begin{eqnarray}\label{100vitanick20}
\Omega^{B,+}=B\cap\{\xi(\omega)>0\}, \quad \Omega^{B,-}=(\Omega\setminus B)\cup \{ \xi(\omega) \leq 0\},
\end{eqnarray}
\begin{eqnarray}\label{100vitanick21}
Q(A)=\int\limits_{A}\psi(\omega)dP(\omega), \quad  A \in {\cal F}, \quad P(\{\omega, \psi(\omega)>0\})=1.
\end{eqnarray}
\end{lemma}
\begin{proof} The necessity. Suppose that the condition (\ref{100vitanick13}) is true. Then, for every $B \in G, \ P(B)>0,$ we have
\begin{eqnarray}\label{100vitanick22}
\int\limits_{B}\xi(\omega)\psi(\omega)dP(\omega)=0,
\end{eqnarray}
or,
\begin{eqnarray}\label{100vitanick23}
\int\limits_{B\cap\{\xi(\omega)>0\}}\xi(\omega)\psi(\omega)dP(\omega)=
-\int\limits_{B\cap\{\xi(\omega)\leq 0\}}\xi(\omega)\psi(\omega)dP(\omega).
\end{eqnarray}
From the equality $ P(B)=P(B\cap\{\xi(\omega)>0\})+P(B\cap\{\xi(\omega)\leq 0\})$
and the equalities (\ref{100vitanick21}), (\ref{100vitanick23}), it follows that $P(B\cap\{\xi(\omega)>0\})>0$ and  $P(B\cap\{\xi(\omega)\leq 0\})>0.$
Therefore, the equality (\ref{100vitanick23}) can be written in the form
\begin{eqnarray}\label{100vitanick24}
0< d^B=\int\limits_{\Omega^{B,+}}\zeta^{B,+}(\omega_2)\psi_2(\omega_2)dP(\omega_2)=
\int\limits_{\Omega^{B,-}}\zeta^{B,+}(\omega_1)\psi(\omega_1)dP(\omega_1).
\end{eqnarray}
Define $\alpha_1(\omega_1,\omega_2)$ by the formula (\ref{100vitanick15}) and prove that the formula (\ref{100vitanick14}) coincide with the formula (\ref{100vitanick21}) for all  $ A \in {\cal F}.$
 But, if to substitute  the expression for  $\alpha_1(\omega_1,\omega_2)$ defined  by the formula (\ref{100vitanick15}) into the formula (\ref{100vitanick14}) and to take into account the expression for $d^B,$ we obtain

$$Q(A)=\int\limits_{\Omega^{B,-}}\chi_{A}(\omega_1)\psi_1(\omega_1)dP(\omega_1)+\int\limits_{\Omega^{B,+}}\chi_{A}(\omega_2)\psi_2(\omega_2)dP(\omega_2)=$$
\begin{eqnarray}\label{100vitanick25}
\int\limits_{A\cap\Omega^{B,-}}\psi(\omega)dP(\omega)+\int\limits_{A\cap\Omega^{ B,+}}\psi(\omega)dP(\omega)=\int\limits_{A}\psi(\omega)dP(\omega).
\end{eqnarray}
The last proves the necessity.

The sufficiency. From the equality 
\begin{eqnarray}\label{100vitanick26}
\chi_{B}\xi(\omega)=\zeta^{B,+}(\omega) - \zeta^{B,-}(\omega)
\end{eqnarray}
 for the measure  $Q,$ given by the formula (\ref{100vitanick14}), it follows the equality
\begin{eqnarray}\label{100vitanick27}
E^Q\chi_{B}\xi(\omega)=0, \quad B \in G.
\end{eqnarray}
The last means that $E^Q\{\xi(\omega)|G\}=0.$
Lemma \ref{101vitanick13} is proved.
\end{proof}

For further investigations, the next Theorem \ref{nick1} is very important \cite{GoncharNick}.
\begin{thm}\label{nick1}
The necessary and sufficient conditions of the local regularity of the nonnegative super-martingale $\{f_m, {\cal F}_m\}_{m=0}^\infty$ relative to a  set of equivalent measures $M$ are the existence of ${\cal F}_m$-measurable random values $\xi_m^0 \in A_0, \  m=\overline{1, \infty},$ such that
\begin{eqnarray}\label{nick2}
\frac{f_m}{f_{m-1}} \leq \xi_m^0, \quad E^P\{\xi_m^0|{\cal F}_{m-1}\}=1, \quad P\in M, \quad m=\overline{1, \infty}.
\end{eqnarray}
\end{thm}
\begin{proof} The necessity.  Without loss of generality, we assume that $f_m\geq a$ for a certain real number $a>0.$ Really, if it is not so, then we can come to the consideration of the super-martingale  $\{f_m+a, {\cal F}_m\}_{m=0}^\infty.$  Thus,   let  $\{f_m, {\cal F}_m\}_{m=0}^\infty$ be a   nonnegative  local regular super-martingale. Then, there exists a nonnegative adapted random process
$\{g_m\}_{m=0}^\infty, \ g_0=0,$ such that $\sup\limits_{P \in M}E^Pg_m<\infty,$
\begin{eqnarray}\label{nick3}
f_{m-1} - E^P\{f_m|{\cal F}_{m-1}\} =E^P\{g_m|{\cal F}_{m-1}\}, \quad P \in M, \quad m=\overline{1, \infty}.
\end{eqnarray}
Let us put $\xi_m^0=\frac{f_m+g_m}{f_{m-1}}, \ m=\overline{1, \infty}.$ Then,  $\xi_m^0 \in A_0 $ and 
from the equalities  (\ref{nick3}) we obtain $E^P\{\xi_m^0|{\cal F}_{m-1}\}=1, \ P \in M, \ m=\overline{1, \infty}.$
It is evident that the inequalities (\ref{nick2}) are valid.

 The sufficiency. Suppose that the conditions of  Theorem \ref{nick1} are valid.
Then,   $ f_m \leq f_{m-1}+ f_{m-1}(\xi_m^0 -1).$
Introduce the  denotation  $g_m= -f_m+ f_{m-1}\xi_m^0.$ Then,  $g_m \geq 0, $  $\sup\limits_{P \in M}E^Pg_m \leq \sup\limits_{P \in M}E^Pf_m + \sup\limits_{P \in M}E^Pf_{m-1}<\infty, \ m=\overline{1, \infty}.$  The last  equality  and the inequalities  give
\begin{eqnarray}\label{nick4}
f_m=f_0+\sum\limits_{i=1}^m f_{i-1}(\xi_i^0 -1) - \sum\limits_{i=1}^m g_{i}, \quad \ m=\overline{1, \infty}.
\end{eqnarray}
Let us consider  the random process  $\{M_m, {\cal F}_m\}_{m=0}^\infty,$ where $M_m=f_0+\sum\limits_{i=1}^m f_{i-1}(\xi_i^0 -1).$ Then,  $E^P\{M_m| {\cal F}_{m-1}\}$ $=M_{m-1}, \ P\in M, \ m=\overline{1, \infty}.$
Theorem \ref{nick1} is proved.
\end{proof}

\section{Completeness of the regular set of  measures.}

In the next two Lemma, we investigate the closure of  a   convex set of equivalent measures presented in Lemma \ref{vitanick9} by the formula (\ref{vitanick11}) that play the fundamental role in the definition of the completeness of the regular  set of measures. First, we consider the countable case.

Suppose that $\Omega_1 $ contains the  countable set  of elementary events and   let ${\cal F}_1$ be a $\sigma$-algebra of all subsets of the set $\Omega_1.$ Let $P_1$  be a measure on the $\sigma$-algebra ${\cal F}_1.$ We assume that $P_1(\omega_i)=p_i>0, \  i=\overline{1, \infty}.$
 On the probability space $\{\Omega_1, {\cal F}_1, P_1\}, $ let us consider a nonnegative random value $\xi_1,$ satisfying the conditions
$$0< P_1(\{\omega \in \Omega_1, \ \eta_1(\omega) <0\})<1,  \quad 0< P_1(\{\omega \in \Omega_1, \ \eta_1(\omega) >0\}),$$
\begin{eqnarray}\label{vitanick22} 
E^{P_1}|\eta_1(\omega)|<\infty,
\end{eqnarray}
where we introduced the denotation $ \eta_1(\omega)=\xi_1(\omega) - 1. $
On the measurable space $\{\Omega_1, {\cal F}_1\},$ let us consider  the set of measures $M_1,$  which are  equivalent to the measure $P_1$ and are given by the formula
$$Q(A)=\sum\limits_{\omega_1 \in  \Omega_1^-} \sum\limits_{\omega_2 \in \Omega_1^+}\chi_{A}(\omega_1) \alpha(\omega_1, \omega_2)\frac{\eta_1^+(\omega_2)}{\eta_1^-(\omega_1)+\eta_1^+(\omega_2)}P_1(\omega_1) P_1(\omega_2)+$$
\begin{eqnarray}\label{vitanick23} 
\sum\limits_{\omega_1 \in \Omega_1^-} \sum\limits_{\omega_2 \in \Omega_1^+} \chi_{A}(\omega_2)\alpha(\omega_1,\omega_2)\frac{\eta_1^-(\omega_1)}{\eta_1^-(\omega_1)+\eta_1^+(\omega_2)}P_1(\omega_1) P_1(\omega_2), \quad A \in {\cal F}_1,
\end{eqnarray}
 where $\eta(\omega)=\eta_1^+(\omega) -\eta_1^-(\omega),$ \ $\Omega_1^+=\{\omega, \eta_1(\omega)>0\}, \ \Omega_1^-=\{\omega, \eta_1(\omega)\leq0\}.$
Introduce the denotations   ${\cal F}_1^+=\Omega_1^+\cap{\cal F}_1,$  ${\cal F}_1^-=\Omega_1^-\cap{\cal F}_1.$  Let  $P_1^-$ be a contraction of the measure $P_1$ on the $\sigma$-algebra ${\cal F}_1^-$   and let  $P_1^+$ be a contraction of the measure $P_1$ on the $\sigma$-algebra ${\cal F}_1^+.$
On the probability space $\{  \Omega_1^- \times \Omega_1^+, {\cal F}_1^-\times {\cal F}_1^+,  P_1^-\times P_1^+\},$ the set of  random value  $\alpha(\omega_1, \omega_2)$  satisfy the conditions
\begin{eqnarray}\label{vitanick24}
P_1\times P_1(\{(\omega_1,\omega_2) \in \Omega_1^- \times \Omega_1^+, \    \alpha(\omega_1, \omega_2)>0\})=P_1(\Omega_1^+)P_1(\Omega_1^-), 
\end{eqnarray}
\begin{eqnarray}\label{vitanick25}  \sum\limits_{\omega_1 \in \Omega^-}\sum\limits_{\omega_2\ \in \Omega^+}\alpha(\omega_1, \omega_2)\frac{\eta_1^-(\omega_1)\eta_1^+(\omega_2)}{\eta_1^-(\omega_1)+\eta_1^+(\omega_2)}P_1(\omega_1)
P_1(\omega_2)<\infty,
\end{eqnarray}
\begin{eqnarray}\label{vitanick26}
\sum\limits_{\omega_1 \in \Omega_1^-}\sum\limits_{\omega_2 \in \Omega_1^+}\alpha(\omega_1, \omega_2)P_1(\omega_1)
P_1(\omega_2)=1.
\end{eqnarray}
On the probability space $\{  \Omega_1^- \times \Omega_1^+, {\cal F}_1^-\times {\cal F}_1^+,  P_1^-\times P_1^+\},$ all  the bounded random values $\alpha(\omega_1, \omega_2)$  the above conditions satisfy.
Introduce into the set of all measures on $\{\Omega_1, {\cal F}_1\}$ the metrics
\begin{eqnarray}\label{vitanick27}
\rho(Q_1, Q_2)=\sum\limits_{i=1}^\infty|Q_1(\omega_i) - Q_2(\omega_i)|.
\end{eqnarray}
\begin{lemma}\label{vitanick28}
 The closure of the set of measures $M_1$  in metrics (\ref{vitanick27}) contains the set of measures
\begin{eqnarray}\label{vitanick29}
\mu_{\omega_1, \omega_2}(A)=\chi_{A}(\omega_1) \frac{\eta_1^+(\omega_2)}{\eta_1^-(\omega_1) +\eta_1^+(\omega_2)}+\chi_{A}(\omega_2) \frac{\eta_1^-(\omega_1)}{\eta_1^-(\omega_1) +\eta_1^+(\omega_2)}
\end{eqnarray}
for   $\omega_1 \in \Omega_1^-,$  $ \omega_2 \in \Omega_1^+,$ $A\in {\cal F}_1.$ 
For every bounded random value $f(\omega),$ the closure of the set of points $E^Qf, \ Q \in M_1,$ in metrics $\rho(x,y)=|x-y|, \ x, y \in R^1,$ contains the points 
$E^{\mu_{\omega_1, \omega_2}}f, \  (\omega_1, \omega_2) \in \Omega_1^-\times \Omega_1^+.$
\end{lemma}
\begin{proof}
Let us choose the set of equivalent measures $Q^{\varepsilon} $ defined by  $\alpha^{\varepsilon}(\omega_1, \omega_2),  0< \varepsilon<1,$ and  given by the law:

$$\alpha^{\varepsilon}(\omega_1^0, \omega_2^0)=\frac{1-\varepsilon}{P_1(\omega_1^0) P_1(\omega_2^0)}, \quad \omega_1^0 \in \Omega_1^-, \quad \omega_2^0 \in \Omega_1^+,$$ 
$$ \alpha^{\varepsilon}(\omega_1, \omega_2)=\varepsilon \alpha_0^{\varepsilon}(\omega_1, \omega_2),     \quad  \alpha_0^{\varepsilon}(\omega_1, \omega_2)=\frac{1}{\sum\limits_{\omega_1\neq \omega_1^0}\sum\limits_{ \omega_\neq \omega_2^0} P(\omega_1)P(\omega_2)},  \quad (\omega_1, \omega_2)\neq (\omega_1^0, \omega_2^0),$$ $$ \quad \omega_1 \in \Omega_1^-, \quad \omega_2 \in \Omega_1^+.$$ 
It is evident that  $  \alpha^{\varepsilon}(\omega_1, \omega_2)>0, (\omega_1, \omega_2) \in \Omega_1^-\times \Omega_1^+,$ for every $1>\varepsilon>0,$ and satisfy the equality
\begin{eqnarray}\label{vitanick31}
\sum\limits_{(\omega_1, \omega_2) \in \Omega_1^-\times \Omega_1^+ }\alpha^{\varepsilon}(\omega_1, \omega_2)P_1(\omega_1)P_1(\omega_2)=1.
\end{eqnarray}
Then,
\begin{eqnarray}\label{vitanick32}
Q^\varepsilon(\omega_1^0)= \sum\limits_{\omega_2 \in \Omega_1^+} \alpha^\varepsilon(\omega_1^0, \omega_2)\frac{\eta_1^+(\omega_2)}{\eta_1^-(\omega_1^0)+\eta_1^+(\omega_2)}P_1(\omega_1^0) P_1(\omega_2),
\end{eqnarray}
\begin{eqnarray}\label{vitanick33} 
Q^\varepsilon(\omega_2^0)= \sum\limits_{\omega_1 \in \Omega_1^-} \alpha^\varepsilon(\omega_1,\omega_2^0)\frac{\eta_1^-(\omega_1)}{\eta_1^-(\omega_1)+\eta_1^+(\omega_2^0)}P_1(\omega_1) P_1(\omega_2^0). 
\end{eqnarray}

$$Q^\varepsilon(\omega_1^0)=(1-\varepsilon)\frac{\eta_1^+(\omega_2^0)}{\eta_1^-(\omega_1^0)+\eta_1^+(\omega_2^0)} +$$
\begin{eqnarray}\label{vitanick34}
 \varepsilon \sum\limits_{\omega_2 \in \Omega_1^+, \omega_2 \neq \omega_2 ^0 } \alpha_0^\varepsilon(\omega_1^0, \omega_2)\frac{\eta_1^+(\omega_2)}{\eta_1^-(\omega_1^0)+\eta_1^+(\omega_2)}P_1(\omega_1^0) P_1(\omega_2),
\end{eqnarray}
 
$$Q^\varepsilon(\omega_2^0)= (1-\varepsilon)\frac{\eta_1^+(\omega_2^0)}{\eta_1^-(\omega_1^0)+\eta_1^+(\omega_2^0)} +$$
\begin{eqnarray}\label{vitanick35}
 \varepsilon \sum\limits_{\omega_1 \in \Omega_1^-, \omega_1\neq \omega_1^0} \alpha_0^\varepsilon(\omega_1,\omega_2^0)\frac{\eta_1^-(\omega_1)}{\eta_1^-(\omega_1)+\eta_1^+(\omega_2^0)}P_1(\omega_1) P_1(\omega_2^0). 
\end{eqnarray}

If $\omega_1\neq \omega_1^0, \omega_2\neq \omega_2^0,$ then

\begin{eqnarray}\label{vitanick32}
Q^\varepsilon(\omega_1)= \varepsilon\sum\limits_{\omega_2 \in \Omega_1^+} \alpha_0^\varepsilon(\omega_1, \omega_2)\frac{\eta_1^+(\omega_2)}{\eta_1^-(\omega_1^0)+\eta_1^+(\omega_2)}P_1(\omega_1) P_1(\omega_2),
\end{eqnarray}
\begin{eqnarray}\label{vitanick33} 
Q^\varepsilon(\omega_2)=\varepsilon \sum\limits_{\omega_1 \in \Omega_1^-} \alpha_0^\varepsilon(\omega_1,\omega_2)\frac{\eta_1^-(\omega_1)}{\eta_1^-(\omega_1)+\eta_1^+(\omega_2)}P_1(\omega_1) P_1(\omega_2). 
\end{eqnarray}

The distance between the measures $Q^\varepsilon$ and $\mu_{\omega_1^0,\omega_2^0}$ is given by the formula

$$\rho(Q^\varepsilon, \mu_{\omega_1^0,\omega_2^0})= \varepsilon+$$

$$ \varepsilon \sum\limits_{\omega_2 \in \Omega_1^+, \omega_2 \neq \omega_2 ^0 } \alpha_0^\varepsilon(\omega_1^0, \omega_2)\frac{\eta_1^+(\omega_2)}{\eta_1^-(\omega_1^0)+\eta_1^+(\omega_2)}P_1(\omega_1^0) P_1(\omega_2)+$$
$$ \varepsilon \sum\limits_{\omega_1 \in \Omega_1^-, \omega_1\neq \omega_1^0} \alpha_0^\varepsilon(\omega_1,\omega_2^0)\frac{\eta_1^-(\omega_1)}{\eta_1^-(\omega_1)+\eta_1^+(\omega_2^0)}P_1(\omega_1) P_1(\omega_2^0)+$$
$$ \varepsilon\sum\limits_{\omega_1 \in \Omega_1^-, \omega_1\neq \omega_1^0}\sum\limits_{\omega_2 \in \Omega_1^+} \alpha_0^\varepsilon(\omega_1, \omega_2)\frac{\eta_1^+(\omega_2)}{\eta_1^-(\omega_1^0)+\eta_1^+(\omega_2)}P_1(\omega_1) P_1(\omega_2)+$$
\begin{eqnarray}\label{vitanick34}
\varepsilon \sum\limits_{\omega_2 \in \Omega_1^+, \omega_2 \neq \omega_2^0}\sum\limits_{\omega_1 \in \Omega_1^-} \alpha_0^\varepsilon(\omega_1,\omega_2)\frac{\eta_1^-(\omega_1)}{\eta_1^-(\omega_1)+\eta_1^+(\omega_2)}P_1(\omega_1) P_1(\omega_2).  
\end{eqnarray}
Since
$$ \sum\limits_{\omega_2 \in \Omega_1^+, \omega_2 \neq \omega_2 ^0 } \alpha_0^\varepsilon(\omega_1^0, \omega_2)\frac{\eta_1^+(\omega_2)}{\eta_1^-(\omega_1^0)+\eta_1^+(\omega_2)}P_1(\omega_1^0) P_1(\omega_2)+$$
$$  \sum\limits_{\omega_1 \in \Omega_1^-, \omega_1\neq \omega_1^0} \alpha_0^\varepsilon(\omega_1,\omega_2^0)\frac{\eta_1^-(\omega_1)}{\eta_1^-(\omega_1)+\eta_1^+(\omega_2^0)}P_1(\omega_1) P_1(\omega_2^0)\leq 1,$$
$$ \sum\limits_{\omega_1 \in \Omega_1^-, \omega_1\neq \omega_1^0}\sum\limits_{\omega_2 \in \Omega_1^+} \alpha_0^\varepsilon(\omega_1, \omega_2)\frac{\eta_1^+(\omega_2)}{\eta_1^-(\omega_1^0)+\eta_1^+(\omega_2)}P_1(\omega_1) P_1(\omega_2)\leq 1,$$
$$  \sum\limits_{\omega_2 \in \Omega_1^+, \omega_2 \neq \omega_2^0}\sum\limits_{\omega_1 \in \Omega_1^-} \alpha_0^\varepsilon(\omega_1,\omega_2)\frac{\eta_1^-(\omega_1)}{\eta_1^-(\omega_1)+\eta_1^+(\omega_2)}P_1(\omega_1) P_1(\omega_2)\leq 1, $$
 we obtain
 $$\rho(Q^\varepsilon, \mu_{\omega_1^0,\omega_2^0}) \leq 4\varepsilon.$$
Let us prove the second part of Lemma \ref{vitanick28}.
It is evident that the inequality
\begin{eqnarray}\label{nicvick1} 
|E^{Q^\varepsilon}f - E^{\mu_{\omega_1, \omega_2}}f| \leq 4 \varepsilon \sup\limits_{\omega \in \Omega_1}|f(\omega)|
\end{eqnarray}
is true. Due to arbitrariness of the  small $\varepsilon,$ Lemma \ref{vitanick28} is proved.
\end{proof}

\begin{defin}
Let $\{\Omega_1, {\cal F}_1\}$ be a measurable space. The  decomposition $A_{n,k}, \  n,k=\overline{1, \infty},$ of the space $\Omega_1$  we call exhaustive one if the following conditions are valid:\\
1) $A_{n,k} \in {\cal F}_1,$ \ $ A_{n,k}\cap A_{n,s}=\emptyset, \ k\neq s,$ \ $\bigcup\limits_{k=1}^\infty A_{n,k}=\Omega_1, \ n=\overline{1, \infty};$\\
2) the $(n+1)$-th decomposition is a sub-decomposition of the $n$-th one, that is, for every $j,$ $A_{n+1,j} \subseteq  A_{n,k}$ for a certain $k=k(j);$\\
3) the minimal $\sigma$-algebra containing all $A_{n,k}, \ n,k=\overline{1, \infty},$ coincides with ${\cal F}_1.$
\end{defin}
The next Remark \ref{g1} is important for the construction of the filtration having the exhaustive decomposition.
\begin{remark}\label{g1}
Suppose that the measurable spaces $\{\Omega_1, {\cal F}_1\}$ and  $\{\Omega_2, {\cal F}_2\}$ have the exhaustive decompositions $A_{n,k}^1, \  n,k=\overline{1, \infty},$ and  $A_{m,s}^2, \  m,s=\overline{1, \infty},$ then the measurable space
 $\{\Omega_1\times \Omega_2, {\cal F}_1\times {\cal F}_2\}$ also have the exhaustive decomposition $B_{n, ks} , \  n=\overline{1, \infty},  k,s=\overline{1, \infty},$ 
$B_{n, ks}=A_{n,k}^1\times A_{n,s}^2, \ k,s=\overline{1, \infty}, \   n=\overline{1, \infty}. $ Really, \\
1) $A_{n,k}^1\times A_{n,s}^2 \in {\cal F}_1\times {\cal F}_2, \ A_{n,k}^1\times A_{n,s}^2\cap A_{n,t}^1\times A_{n,r}^2=\emptyset, \ (k, s)\neq (t,r),$\\
$\bigcup\limits_{k,s=1}^\infty B_{n,ks}=\Omega_1\times \Omega_2, \ n=\overline{1, \infty};$\\
2) the $(n+1)$-th  decomposition is a sub-decomposition of the $n$-th one, that is, for every $k,s$ \  $B_{n+1,ks} \subseteq  B_{n,ij}$ for a certain $i=i(k), j= j(s)$;\\
3)  the minimal $\sigma$-algebra containing all $B_{n,ks}, \ n,k,s=\overline{1, \infty},$ coincides with ${\cal F}_1\times {\cal F}_2.$
\end{remark}

\begin{lemma}\label{2vitanick1}
Let a measurable space $\{\Omega, {\cal F}\}$ have an exhaustive decomposition and let  $\xi$ be an integrable random value relative to the measure $P,$ satisfying the conditions (\ref{vitanick10}). Then, the closure of the set of measure $Q,$ given by the formula  (\ref{vitanick11}), relative to the pointwise  convergence of measures  contains the set of measures 

$$\nu_{(\omega_1,\omega_2)}(A)=\chi_{A}(\omega_1)\frac{\xi^+(\omega_2)}{\xi^-(\omega_1)+\xi^+(\omega_2)}+$$
\begin{eqnarray}\label{2vitanick2}
\chi_{A}(\omega_2)\frac{\xi^-(\omega_1)}{\xi^-(\omega_1)+\xi^+(\omega_2)},\quad A \in {\cal F}, \quad (\omega_1,\omega_2) \in  \Omega^-\times  \Omega^+,
\end{eqnarray}
 for those $(\omega_1, \omega_2) \in   \Omega^-\times  \Omega^+$ which have the full measure $\mu= P^-\times P^+.$
For every integrable finite valued  random value $f(\omega)$ relative to all measures $Q,$  the closure in metrics $\rho(x_1,x_2)=|x_1-x_2|, $ $x_1, x_2 \in R^1,$ of the set of  real numbers
$$\int\limits_{\Omega^-}\int\limits_{\Omega^+}f(\omega_1)\alpha(\omega_1, \omega_2)\frac{\xi^+(\omega_2)}{\xi^-(\omega_1)+\xi^+(\omega_2)}d\mu(\omega_1,\omega_2)+$$
\begin{eqnarray}\label{2vitanick3} \int\limits_{\Omega^-}\int\limits_{\Omega^+}f(\omega_2)\alpha(\omega_1, \omega_2)\frac{\xi^-(\omega_1)}{\xi^-(\omega_1)+\xi^+(\omega_2)}d\mu(\omega_1,\omega_2),
\end{eqnarray}
when $\alpha(\omega_1, \omega_2)$ runs over all random values satisfying the conditions
(\ref{vitanick12}), (\ref{vitanick14}), contains the set of numbers
$$f(\omega_1)\frac{\xi^+(\omega_2)}{\xi^-(\omega_1)+\xi^+(\omega_2)}+$$
\begin{eqnarray}\label{2vitanick4}
f(\omega_2)\frac{\xi^-(\omega_1)}{\xi^-(\omega_1)+\xi^+(\omega_2)},\ (\omega_1,\omega_2) \in  \Omega^-\times  \Omega^+.
\end{eqnarray}
\end{lemma}
\begin{proof}
On a probability space   $\{\Omega, {\cal F}, P\},$ 
let $\xi$ be an integrable  random value, satisfying the conditions (\ref{vitanick10}). As before,
let $\Omega^+=\{\omega, \xi(\omega)>0\}, \ \Omega^-=\{\omega, \xi(\omega)\leq 0\}$ and let ${\cal F}^-,$ $  {\cal F}^+$ be the restrictions of the $\sigma$-algebra ${\cal F}$ on the sets  $\Omega^-$ and  $\Omega^+,$ correspondingly. Suppose that $P^-$ and $P^+$ are the contractions of the measure $P$ on the $\sigma$-algebras ${\cal F}^-,$ $  {\cal F}^+,$ correspondingly.
Consider the probability space   $\{  \Omega^-\times  \Omega^+, {\cal F}^-\times{\cal F}^+, P^-\times P^+\}$ which is a direct product of the probability spaces  $\{\Omega^-, {\cal F}^-, P^-\}$ and $\{\Omega^+, {\cal F}^+, P^+\}.$ Due to Lemma \ref{2vitanick1} and Remark \ref{g1}, the measurable space   $\{\Omega^-\times  \Omega^+, {\cal F}^-\times{\cal F}^+\}$ has the exhaustive decomposition $B_{n,ks}, \  k,s=\overline{1,\infty}, \ n=\overline{1,\infty}. $ Denote ${\cal F}_n$ the minimal $\sigma$-algebra generated by decomposition $B_{n,ks}, \  k,s=\overline{1,\infty}. $ It is evident that ${\cal F}_n \subset {\cal F}_{n+1}.$ Moreover, $\sigma(\bigvee\limits_{n=1}^\infty {\cal F}_n)={\cal F}^-\times{\cal F}^+.$  On the probability space  
$\{  \Omega^-\times  \Omega^+, {\cal F}^-\times{\cal F}^+, P^-\times P^+\},$   for every integrable finite valued random value $ f(\omega_1,\omega_2)$  the sequence 
$E^{\mu}\{ f(\omega_1,\omega_2)|{\cal F}_n\}$ converges to $f(\omega_1,\omega_2)$ with probability one, as $n \to \infty,$ since it is a regular martingale. It is evident that for those $B_{n,ks}$ for which $\mu(B_{n,ks})\neq 0$
\begin{eqnarray}\label{2vitanick5}
E^{\mu}\{ f(\omega_1,\omega_2)|{\cal F}_n\}=\frac{\int\limits_{B_{n,ks}} f(\omega_1,\omega_2)d\mu}{\mu(B_{n,ks})}, \quad (\omega_1,\omega_2) \in B_{n,ks}.
\end{eqnarray}
Denote $D_0=\bigcup\limits_{n, k, s, \mu(B_{n,ks})=0}B_{n,ks}.$ It is evident that 
$\mu(D_0)=0.$ For every $(\omega_1,\omega_2) \in  \Omega^-\times  \Omega^+\setminus D_0,$ the formula (\ref{2vitanick5}) is well defined and is finite.
Let $D_1$ be the subset of the set $ \Omega^-\times  \Omega^+\setminus D_0,$ where the limit of the left hand  side of the formula (\ref{2vitanick5}) does not exists. Then, $\mu(D_1)=0.$ For every  $(\omega_1,\omega_2) \in  \Omega^-\times  \Omega^+\setminus (D_0 \cup D_1),$ the right hand side of the formula (\ref{2vitanick5}) converges to  $f(\omega_1,\omega_2).$
For   $(\omega_1,\omega_2) \in  \Omega^-\times  \Omega^+\setminus (D_0 \cup D_1),$ denote $A_n=A_n(\omega_1, \omega_2)$ those set  $B_{n,ks}$   for which  $(\omega_1,\omega_2) \in B_{n,ks}$ for a certain $k,s.$ Then, for every integrable finite valued $f(\omega_1,\omega_2)$
\begin{eqnarray}\label{2vitanick6}
\lim\limits_{n \to \infty} \frac{\int\limits_{A_n} f(\omega_1,\omega_2)d\mu}{\mu(A_n)}=f(\omega_1,\omega_2).
\end{eqnarray}
Let us consider the sequence
\begin{eqnarray}\label{2vitanick7}
\alpha_n^{\varepsilon_n}(\omega_1,\omega_2)=(1-\varepsilon_n)\frac{\chi_{A_n}(\omega_1,\omega_2)}{\mu(A_n)}+ \varepsilon_n\frac{\chi_{\Omega^-\times  \Omega^+\setminus A_n}(\omega_1,\omega_2)}{\mu(\Omega^-\times  \Omega^+\setminus A_n)},
\end{eqnarray}
where $0<\varepsilon_n<1,\  \lim\limits_{n \to \infty}\varepsilon_n=0.$
Such a sequence $\alpha_n^{\varepsilon_n}(\omega_1,\omega_2)$ satisfy the conditions
(\ref{vitanick12}) - (\ref{vitanick14}) and 

$$Q_n^{\varepsilon_n}(A)=\int\limits_{\Omega^-}\int\limits_{\Omega^+}\chi_{A}(\omega_1)\alpha_n^{\varepsilon_n}(\omega_1, \omega_2)\frac{\xi^+(\omega_2)}{\xi^-(\omega_1)+\xi^+(\omega_2)}d\mu(\omega_1,\omega_2)+$$

$$ \int\limits_{\Omega^-}\int\limits_{\Omega^+}\chi_{A}(\omega_2)\alpha_n^{\varepsilon_n}(\omega_1, \omega_2)\frac{\xi^-(\omega_1)}{\xi^-(\omega_1)+\xi^+(\omega_2)}d\mu(\omega_1,\omega_2)=$$

$$(1-\varepsilon_n)\frac{\int\limits_{A_n}\chi_{A}(\omega_1)\frac{\xi^+(\omega_2)}{\xi^-(\omega_1)+\xi^+(\omega_2)}d\mu(\omega_1,\omega_2)}{\mu(A_n)}+$$

 $$(1-\varepsilon_n)\frac{\int\limits_{A_n}\chi_{A}(\omega_2)\frac{\xi^-(\omega_1)}{\xi^-(\omega_1)+\xi^+(\omega_2)}d\mu(\omega_1,\omega_2)}{\mu(A_n)}+$$

$$\varepsilon_n\frac{\int\limits_{\Omega^-\times \Omega^+\setminus A_n}\chi_{A}(\omega_1)\frac{\xi^+(\omega_2)}{\xi^-(\omega_1)+\xi^+(\omega_2)}d\mu(\omega_1,\omega_2)}{\mu(\Omega^-\times \Omega^+\setminus A_n)}+$$
\begin{eqnarray}\label{2vitanick8}
 \varepsilon_n\frac{\int\limits_{\Omega^-\times \Omega^+\setminus A_n}\chi_{A}(\omega_2)\frac{\xi^-(\omega_1)}{\xi^-(\omega_1)+\xi^+(\omega_2)}d\mu(\omega_1,\omega_2)}{\mu(\Omega^-\times \Omega^+\setminus A_n)}.
\end{eqnarray}
From the  formula (\ref{2vitanick8}), we obtain

$$\lim\limits_{n \to \infty}Q_n^{\varepsilon_n}(A)=\chi_{A}(\omega_1)\frac{\xi^+(\omega_2)}{\xi^-(\omega_1)+\xi^+(\omega_2)}+$$
\begin{eqnarray}\label{2vitanick9}
\chi_{A}(\omega_2)\frac{\xi^-(\omega_1)}{\xi^-(\omega_1)+\xi^+(\omega_2)},\ A \in {\cal F}, \ (\omega_1,\omega_2) \in  \Omega^-\times  \Omega^+\setminus (D_0  \cup D_1).
\end{eqnarray}
Further,
$$E^{Q_n^{\varepsilon_n}}f(\omega)=\int\limits_{\Omega^-}\int\limits_{\Omega^+}
f(\omega_1)\alpha_n^{\varepsilon_n}(\omega_1, \omega_2)\frac{\xi^+(\omega_2)}{\xi^-(\omega_1)+\xi^+(\omega_2)}d\mu(\omega_1,\omega_2)+$$
$$ \int\limits_{\Omega^-}\int\limits_{\Omega^+}f(\omega_2)\alpha_n^{\varepsilon_n}(\omega_1, \omega_2)\frac{\xi^-(\omega_1)}{\xi^-(\omega_1)+\xi^+(\omega_2)}d\mu(\omega_1,\omega_2)=$$

$$(1-\varepsilon_n)\frac{\int\limits_{A_n}f(\omega_1)\frac{\xi^+(\omega_2)}{\xi^-(\omega_1)+\xi^+(\omega_2)}d\mu(\omega_1,\omega_2)}{\mu(A_n)}+$$

 $$(1-\varepsilon_n)\frac{\int\limits_{A_n}f(\omega_2)\frac{\xi^-(\omega_1)}{\xi^-(\omega_1)+\xi^+(\omega_2)}d\mu(\omega_1,\omega_2)}{\mu(A_n)}+$$

$$\varepsilon_n\frac{\int\limits_{\Omega^-\times \Omega^+\setminus A_n}f(\omega_1)\frac{\xi^+(\omega_2)}{\xi^-(\omega_1)+\xi^+(\omega_2)}d\mu(\omega_1,\omega_2)}{\mu(\Omega^-\times \Omega^+\setminus A_n)}+$$
\begin{eqnarray}\label{3vitanick9}
 \varepsilon_n\frac{\int\limits_{\Omega^-\times \Omega^+\setminus A_n}f(\omega_2)\frac{\xi^-(\omega_1)}{\xi^-(\omega_1)+\xi^+(\omega_2)}d\mu(\omega_1,\omega_2)}{\mu(\Omega^-\times \Omega^+\setminus A_n)}.
\end{eqnarray}
From the formula (\ref{3vitanick9}), we obtain

$$\lim\limits_{n \to \infty}E^{Q_n^{\varepsilon_n}}f(\omega)=f(\omega_1)\frac{\xi^+(\omega_2)}{\xi^-(\omega_1)+\xi^+(\omega_2)}+$$
\begin{eqnarray}\label{2vitanick9}
f(\omega_2)\frac{\xi^-(\omega_1)}{\xi^-(\omega_1)+\xi^+(\omega_2)},\ A \in {\cal F}, \ (\omega_1,\omega_2) \in  \Omega^-\times  \Omega^+\setminus (D_0  \cup D_1).
\end{eqnarray}
Lemma \ref{2vitanick1} is proved.
\end{proof}

The next Theorem \ref{koljavita1} is a consequence of Lemma \ref{vitanick9}.

\begin{thm}\label{koljavita1}
On the probability space   $\{\Omega, {\cal F}, P\},$ for the nonnegative random value 
$\xi\neq 1$   the set of measures $M_0$ on the measurable space   $\{\Omega, {\cal F}\},$ being equivalent to the measure $P,$  satisfies the condition
\begin{eqnarray}\label{koljavita2}
E^Q\xi=1, \quad Q \in M_0,
\end{eqnarray}
 if and only if
as for $Q \in M_0$ the representation
$$Q(A)=\int\limits_{\Omega^-}\int\limits_{\Omega^+}\chi_{A}(\omega_1)\alpha(\omega_1, \omega_2)\frac{(\xi-1)^+(\omega_2)}{(\xi-1)^-(\omega_1)+(\xi-1)^+(\omega_2)}d\mu(\omega_1,\omega_2)+$$
\begin{eqnarray}\label{koljavita3} \int\limits_{\Omega^-}\int\limits_{\Omega^+}\chi_{A}(\omega_2)\alpha(\omega_1, \omega_2)\frac{(\xi-1)^-(\omega_1)}{(\xi-1)^-(\omega_1)+(\xi-1)^+(\omega_2)}d\mu(\omega_1,\omega_2), \quad A \in {\cal F},
\end{eqnarray}
 is true, where on the measurable space $\{  \Omega^-\times  \Omega^+, {\cal F}^-\times{\cal F}^+, P^-\times P^+\},$  the random value $\alpha(\omega_1, \omega_2)$  satisfies the conditions
\begin{eqnarray}\label{koljavita4}
\mu(\{(\omega_1,\omega_2) \in \Omega^- \times \Omega^+, \    \alpha(\omega_1, \omega_2)>0\})=P(\Omega^+)P(\Omega^-), 
\end{eqnarray}
\begin{eqnarray}\label{koljavita5}  \int\limits_{\Omega^-}\int\limits_{\Omega^+}\alpha(\omega_1, \omega_2)\frac{(\xi-1)^-(\omega_1) (\xi-1)^+(\omega_2)}{(\xi-1)^-(\omega_1)+(\xi-1)^+(\omega_2)}d\mu(\omega_1,\omega_2)<\infty,
\end{eqnarray}

\begin{eqnarray}\label{koljavita6}
\int\limits_{\Omega^-}\int\limits_{\Omega^+}\alpha(\omega_1, \omega_2)d\mu(\omega_1,\omega_2)=1.
\end{eqnarray}
\end{thm}
We introduced above the following denotations: $\mu = P^-\times P^+,$ \  $P^-$ is a contraction of the measure $P$ on the set $\Omega^-=\{\omega \in \Omega, \   \xi -1 \leq 0\},$ \ $P^+$ is a contraction of the measure $P$ on the set $\Omega^+=\{\omega \in \Omega, \  \xi -1 > 0\},$ \
${\cal F}^-=\Omega^-\cap {\cal F},$ ${\cal F}^+=\Omega^+\cap {\cal F}.$

It is evident that the set of measure $M_0$ is a nonempty one, since it contains those measures $Q,$ for which the random value $\alpha(\omega_1, \omega_2)$ is bounded, since $E^Q|\xi -1|< \infty.$  

\begin{thm}\label{koljavita7}
On the probability space  $\{\Omega, {\cal F}, P\}$ with the filtration ${\cal F}_n$ on it,  the set of measures $M_0,$ given by the formula (\ref{koljavita3}), is consistent with the filtration ${\cal F}_n,$ if and only if,  as $E^Q\{\xi|{\cal  F}_n\}, Q \in M_0,$ is a local regular martingale.
\end{thm}
\begin{proof}
The necessity. Let the set of measures $M_0$ be consistent with the filtration. Then, due to Theorem \ref{tatnick8},  $E^Q\{\xi|{\cal  F}_n\}, Q \in M_0,$ is a local regular martingale.

The sufficiency. Suppose that $E^Q\{\xi|{\cal  F}_n\}, Q \in M_0,$ is a local regular martingale. Let us prove that,  if $Q_1,\ Q_2 \in M_0,$ then
the set of measures
\begin{eqnarray}\label{koljavita8}
R_s^k(A)=\int\limits_{A}\frac{E^{Q_2}\{\frac{dQ_1}{dQ_2}|{\cal F}_k\}}{E^{Q_2}\{\frac{dQ_1}{dQ_2}|{\cal F}_s\}}dQ_2,  \quad A \in {\cal F},
 \quad k\geq s \geq n, \quad n=\overline{0, \infty},
\end{eqnarray}
belongs to the set $M_0.$ For this, it is to prove that $E^{R_s^k}(\xi -1)=0,$ or
 $E^{R_s^k}\xi =1.$ Really, if $E^{Q_1}\xi=1, \ E^{Q_2}\xi=1,$ then
$$ E^{R_s^k}\xi=E^{Q_2}\xi \frac{E^{Q_2}\{\frac{dQ_1}{dQ_2}|{\cal F}_k\}}{E^{Q_2}\{\frac{dQ_1}{dQ_2}|{\cal F}_s\}}=
E^{Q_2} E^{Q_2}\{\xi|{\cal F}_k\}\frac{\frac{dQ_1}{dQ_2}}{E^{Q_2}\{\frac{dQ_1}{dQ_2}|{\cal F}_s\}}=$$ 
$$ E^{Q_2}E^{Q_2}\{ E^{Q_2}\{\xi|{\cal F}_k\}\frac{\frac{dQ_1}{dQ_2}}{E^{Q_2}\{\frac{dQ_1}{dQ_2}|{\cal F}_s\}}|{\cal F}_s\}=$$

 $$E^{Q_2}E^{Q_1}\{E^{Q_2}\{\xi|{\cal F}_k\}|{\cal F}_s\}=E^{Q_2}E^{Q_1}\{E^{Q_1}\{\xi|{\cal F}_k\}|{\cal F}_s\}=$$
\begin{eqnarray}\label{koljavita9}
E^{Q_2}E^{Q_1}\{\xi|{\cal F}_s\}= E^{Q_2}E^{Q_2}\{\xi|{\cal F}_s\}= E^{Q_2}\xi=1.
\end{eqnarray}
Theorem \ref{koljavita7} is proved.
\end{proof}

\begin{thm}\label{mykolvit1}
On the probability space  $\{\Omega, {\cal F}, P\} $  with the filtration ${\cal F}_n$ on it, the set of measures $M_0,$ given by the formula  (\ref{koljavita3}), is consistent with the filtration ${\cal F}_n,$ if and only if there exists not depending on $(\omega_1, \omega_2) \in \Omega^-\times \Omega^+$ the random process $\{m_n, {\cal F}_n\}_{n=0}^\infty$ such that
\begin{eqnarray}\label{mykolvit2}
E^{\nu_{\omega_1, \omega_2}}\{\xi|{\cal F}_n\}=m_n,\quad n=\overline{1, \infty},
\end{eqnarray} 
  for  those $(\omega_1, \omega_2) \in \Omega^-\times \Omega^+$ that have the full  measure $\mu=P^- \times P^+,$ where
 $$\nu_{\omega_1,\omega_2}(A)=\chi_{A}(\omega_1)\frac{(\xi-1)^+(\omega_2)}{(\xi-1)^-(\omega_1)+(\xi-1)^+(\omega_2)}+$$
\begin{eqnarray}\label{mykolvit3}
\chi_{A}(\omega_2)\frac{(\xi-1)^-(\omega_1)}{(\xi-1)^-(\omega_1)+(\xi-1)^+(\omega_2)}, \quad A \in {\cal F}, \quad (\omega_1,\omega_2) \in \Omega^-\times  \Omega^+.
\end{eqnarray}
\end{thm}
\begin{proof}
The necessity. Suppose that the set of measures $M_0,$ given by the formula (\ref{koljavita3}), is consistent with the filtration ${\cal F}_n.$   Due to Theorem \ref{koljavita7},  $E^Q\{\xi|{\cal  F}_n\}, \ Q \in M_0,$ is  a local regular martingale.
Then, $E^Q\{\xi|{\cal  F}_n\}=m_n.$ Using Lemma \ref{2vitanick1}, we obtain
$E^{\nu_{\omega_1, \omega_2}}\{\xi|{\cal F}_n\}=m_n$ for   those $(\omega_1, \omega_2) \in \Omega^-\times \Omega^+$ that have the full  measure $\mu.$

The sufficiency. If the formula (\ref{mykolvit2}) is true, then  $E^Q\{\xi|{\cal  F}_n\}=m_n, \  Q \in M_0.$ From this, it follows that  $E^Q\{\xi|{\cal  F}_n\}, \ Q \in M_0,$ is  a local regular martingale.
Theorem \ref{mykolvit1} is proved.
\end{proof}

\begin{defin}\label{200vitanick1}
On the probability space $\{\Omega, {\cal F}, P\}$ with the filtration ${\cal F}_n$ on it, the consistent with the filtration ${\cal F}_n$   subset of the measures $M $ of the set of the measures  $M_0$ that is generated by the nonnegative  random value  $\xi \neq 1,$  $ \ E^Q \xi=1, \ Q \in M_0, $ 
we call the regular set of measures.
\end{defin}

Let   $\{\Omega, {\cal F}, P\} $  be a probability space.
On the measurable space  $\{\Omega, {\cal F}\} $ with the filtration ${\cal F}_n$ on it,  let  $M \subseteq M_0 $ be a set  of regular measures, where the set $M_0$ is generated by the nonnegative random value $\xi \neq 1.$  Denote by   $\{m_n, {\cal F}_n\}_{n=0}^\infty$ the regular martingale, where $m_n=E^Q\{\xi| {\cal F}_n\}, Q\in M, \ n=\overline{1,\infty}.$ 
Assume that  $M_n$ is a contraction of the set of regular  measures $M$ onto the $\sigma$-algebra ${\cal F}_n.$ Every $Q^n \in M_n$ is equivalent to $P^n,$ where $P^n$ is a contraction of the measure $P$ on the $\sigma$-algebra ${\cal F}_n.$ 
 For every $Q^n \in M_n,$ we have
 $E^{Q^n}[m_n- m_{n-1}]=0. $ 
Therefore, for the measure $Q_n \in M_n$ the representation
$$Q_n(A)=\int\limits_{\Omega_n^-\times\Omega_n^+}\chi_{A}(\omega_1)\frac{\alpha_n(\omega_1,\omega_2)[m_n-m_{n-1}]^+(\omega_2)}{[m_n-m_{n-1}]^-(\omega_1)+[m_n-m_{n-1}]^+(\omega_2)}d\mu_n(\omega_1,\omega_2)+$$
\begin{eqnarray}\label{77vitanick1}
\int\limits_{\Omega_n^-\times\Omega_n^+}\chi_{A}(\omega_2)\frac{\alpha_n(\omega_1,\omega_2)[m_n-m_{n-1}]^-(\omega_1)}{[m_n-m_{n-1}]^-(\omega_1)+[m_n-m_{n-1}]^+(\omega_2)}d\mu_n(\omega_1,\omega_2), \quad A \in {\cal F}_n,
\end{eqnarray}
$$ \Omega_n^-=\{\omega_1 \in \Omega, \  [m_n-m_{n-1}] (\omega_1)\leq 0\},$$
$$  \Omega_n^+=\{\omega_2 \in \Omega, \  [m_n-m_{n-1}](\omega_2) > 0\},$$
is true, where,  on the measurable space $\{ \Omega_n^-\times  \Omega_n^+, {\cal F}_n^-\times {\cal F}_n^+\},$ the random value $\alpha_n(\omega_1, \omega_2) $  satisfies the conditions
\begin{eqnarray}\label{vitakolja4}
\mu_n(\{(\omega_1,\omega_2) \in \Omega_n^- \times \Omega_n^+, \    \alpha_n(\omega_1, \omega_2)>0\})=P_n(\Omega^+)P_n(\Omega^-), 
\end{eqnarray}
\begin{eqnarray}\label{vitakolja5}  \int\limits_{\Omega_n^-}\int\limits_{\Omega_n^+}\alpha_n(\omega_1, \omega_2)\frac{[m_n-m_{n-1}]^-(\omega_1) [m_n-m_{n-1}]^+(\omega_2)}{[m_n-m_{n-1}]^-(\omega_1)+[m_n-m_{n-1}]^+(\omega_2)}d\mu_n(\omega_1,\omega_2)<\infty,
\end{eqnarray}

\begin{eqnarray}\label{vitakolja6}
\int\limits_{\Omega_n^-}\int\limits_{\Omega_n^+}\alpha_n(\omega_1, \omega_2)d\mu_n(\omega_1,\omega_2)=1.
\end{eqnarray}
Here, the measure 
$\mu_n=P^n_-\times P^n_+$  is given on the measurable space $\{ \Omega_n^-\times  \Omega_n^+, {\cal F}_n^-\times {\cal F}_n^+\}$ and it is a direct product of the measures $ P^n_-$  and  $ P^n_+,$ where the measure $P^n_+$ is a contraction of the measure $P^n$ on the $\sigma$-algebra ${\cal F}_n^+= \Omega_n^+\cap {\cal F}_n$ and 
 $P^n_-$ is a contraction of the measure $P^n$ on the $\sigma$-algebra $ {\cal F}_n^-=\Omega_n^-\cap {\cal F}_n.$

\begin{defin}\label{211vitanick}
 We say that the  regular set of  measures $M$ is complete one, if for every $n=\overline{1,\infty}$  the set of measures $Q_n$ contains the measures of the kind (\ref{77vitanick1}) for the random values $\alpha_n(\omega_1, \omega_2)$   of the kind
$C_1^n\chi_{A}(\omega_1, \omega_2)+C_2^n\chi_{\Omega_n^-\times  \Omega_n^+\setminus A}(\omega_1, \omega_2),$ as $A$ runs all sets from the $\sigma$-algebra $ {\cal F}_n^-\times {\cal F}_n^+,$  where  $C_1^n \mu_n(A)+ C_2^n \mu_n(\Omega_n^-\times  \Omega_n^+\setminus A)=1, \  C_1^n \geq 0, \ C_2^n \geq 0.$
\end{defin}

It is evident that the regular set of measures $M$ is a convex set of measure.

On the probability space $\{\Omega, {\cal F}, P\}$ with the filtration ${\cal F}_n$ on it,
let us  introduce into consideration the set  $A_0$ of all integrable  nonnegative random values $\zeta$  relative to the  set of regular measures $M,$  satisfying the conditions
\begin{eqnarray}\label{0mars6}
E^P\zeta=1, \quad P \in M.
\end{eqnarray}
It is evident that the set  $A_0$ is a nonempty one, since it contains  the random value $\zeta =1.$  The more interesting case is as $A_0$ contains more then one element. So, further  we consider the regular set of measure $M$ with the set $A_0,$ containing more then one element.

The set   $A_0$ can contain more then two elements. Then, for every element $\eta \in A_0$  $E^Q\{\eta|{\cal  F}_n\}, \  Q \in M,$ forms the local regular martingale.

 In the next Lemma  \ref{1vitanick35}, using Lemma \ref{vitanick9},   we construct   a  set of measures consistent with the filtration. 
On the probability space $\{\Omega_1^0, {\cal F}_1^0, P_1\}, $ let us consider a nonnegative random value $\xi_1,$ satisfying the conditions
$$0< P_1(\{\omega \in \Omega_1^0, \ \eta_1(\omega) <0\})<1,   $$
\begin{eqnarray}\label{1vitanick22} 
0< P_1(\{\omega \in \Omega_1^0, \ \eta_1(\omega) >0\}),
\end{eqnarray}
where we introduced the denotation $\eta_1(\omega)=\xi_1(\omega) - 1.$ Described in Lemma \ref{vitanick9}  the set of equivalent measures  to the measure $P_1$ and such that  $E^Q \eta_1(\omega)=0,$ we denote by $M_1.$
Let us construct the infinite direct  product of the measurable spaces $\{\Omega_i^0, {\cal F}_i^0\}, \ i=\overline{1, \infty}, $ where $\Omega_i^0=\Omega_1^0, \  {\cal F}_i^0= {\cal F}_1^0.$
Denote $\Omega=\prod\limits_{i=1}^\infty\Omega_i^0.$ On the space  $\Omega,$ under the $\sigma$-algebra ${\cal F}$  we understand the minimal $\sigma$-algebra, generated by the sets $\prod\limits_{i=1}^\infty G_i, \  G_i \in {\cal F}_i^0,$ where in the last product  only the finite set of $G_i$ do not equal $\Omega_i^0.$
 On the measurable space $\{\Omega, {\cal F}\},$ under the filtration  ${\cal F}_n$ we understand the minimal 
$\sigma$-algebra generated by the sets $\prod\limits_{i=1}^\infty G_i, \  G_i \in {\cal F}_i^0,$ where $G_i=\Omega_i^0$ for  $ i>n.$ We consider the  probability space
 $\{\Omega, {\cal F}, P\}, $ where $P=\prod\limits_{i=1}^\infty P_i, \ P_i=P_1, \ i=\overline{1, \infty}.$

 On the measurable space  $\{\Omega, {\cal F}\},$ we  introduce into consideration the set of measures $M,$ where $Q$ belongs to $ M,$ if $Q=\prod\limits_{i=1}^\infty Q_i, \  Q_i \in M_1.$ We  denote by $M^{Q_0} $ a subset of the set $M$ of those measures  $Q=\prod\limits_{i=1}^\infty Q_i, \  Q_i \in M_1,$ for which only the finite set of $Q_i$ does not coincide with the measure $Q_0 \in M_1.$

\begin{lemma}\label{1vitanick35} 
 On the measurable space  $\{\Omega, {\cal F}\}$ with the filtration ${\cal F}_n$ on it, there exists consistent with the filtration ${\cal F}_n$ the set of measures $M_0$ and the nonnegative random variable $\xi_0$ such that $E^Q\xi_0=1, \ Q \in M_0,$    if the random value $\xi_1,$ satisfying the conditions (\ref{1vitanick22}), is bounded. 
\end{lemma} 
\begin{proof} To prove Lemma \ref{1vitanick35}, we need to construct  a nonnegative bounded random value $\xi_0$ on the measurable  space $\{\Omega, {\cal F}\}$ and a set of equivalent measures $M_0$ on it, such that $E^Q\xi_0=1, \  Q\in M_0,$ and to prove that the set of measures  $M_0 $ is consistent with the filtration ${\cal F}_n.$  From  the Lemma \ref{1vitanick35} conditions,    the random value $\eta_1(\omega_1)=\xi_1(\omega_1) - 1 $ is also bounded. Let us put 
\begin{eqnarray}\label{vitanick36}
\xi_0=\prod\limits_{i=1}^\infty[1+a_i(\omega_1, \ldots, \omega_{i-1})\eta_i(\omega_i)],
\end{eqnarray}
 where the random values $ a_i(\omega_1, \ldots, \omega_{i-1})$  are ${\cal F}_{i-1}$-measurable, $  i=\overline{1,\infty}, $ they satisfy the conditions  $0< a_i(\omega_1, \ldots, \omega_{i-1})\leq b_i<1.$  The  constants $b_i$ are such that $\sum\limits_{i=1}^\infty b_i<\infty,$
the random value $\eta_i(\omega_i)$ is given on $\{\Omega_i^0, {\cal F}_i^0, P_i\}$ and is distributed as $\eta_1(\omega_1)$ on $\{\Omega_1^0, {\cal F}_1^0, P_1\}.$   From this, it follows that the random value $\xi_0$ is bounded by the constant $\prod\limits_{i=1}^\infty[1+C b_i],$ where $C>0$ and it is such that  $|\eta_i(\omega_i)|<C, \  i=\overline{1,\infty}.$
It is evident that $E^Q\xi_0=1, \ Q \in M^{Q_0}.$ Really, 
$$E^Q\prod\limits_{i=1}^n[1+a_i(\omega_1, \ldots, \omega_{i-1})\eta_i(\omega_i)]=$$
$$E^{Q_0^{n-1}}\prod\limits_{i=1}^{n-1}[1+a_i(\omega_1, \ldots, \omega_{i-1})\eta_i(\omega_i)]\times$$
\begin{eqnarray}\label{vitanick37}
E^{Q_n}[1+a_{n-1}(\omega_1, \ldots, \omega_{n-1})\eta_n(\omega_n)],
\end{eqnarray}
where $Q=\prod\limits_{i=1}^{\infty}Q_i, \  Q_0^{n-1}=\prod\limits_{i=1}^{n-1}Q_i,$
$$E^{Q_n}[1+a_{n-1}(\omega_1, \ldots, \omega_{n-1})\eta_i(\omega_n)]=$$
\begin{eqnarray}\label{vitanick38}
[1+a_{n-1}(\omega_1, \ldots, \omega_{n-1})E^{Q_n}\eta_n(\omega_n)]=1.
\end{eqnarray}
From the last equality, we have
\begin{eqnarray}\label{vitanick39}
E^Q\prod\limits_{i=1}^n[1+a_i(\omega_1, \ldots, \omega_{i-1})\eta_i(\omega_i)]=1.
\end{eqnarray}
Since $\xi_0=\lim\limits_{n \to \infty}\prod\limits_{i=1}^n[1+a_i(\omega_1, \ldots, \omega_{i-1})\eta_i(\omega_i)],$ from the equality (\ref{vitanick39}) and the possibility to go to the limit under the mathematical expectation, we prove the needed statement.
Let us prove the existence of the set of measures $M_0$ consistent with the filtration ${\cal F}_n.$
If $Q\in M^{Q_0},$ then
\begin{eqnarray}\label{vitanick40}
E^Q\{\xi_0|{\cal F}_n\}=\prod\limits_{i=1}^n[1+a_i(\omega_1, \ldots, \omega_{i-1})\eta_i(\omega_i)], \quad Q\in M^{Q_0}.
\end{eqnarray}
Due to Lemma  \ref{vitanick1}, there exists a set of measures $M_0$  such that it is consistent with the filtration and $M_0 \supseteq M^{Q_0},$   $ E^Q\xi_0=1, \  Q \in M_0.$  
The set $M_0$ is a linear convex span of the set  $M^{Q_0}.$  It means that the set of measures  $M_0$  is consistent with the filtration.
Lemma \ref{1vitanick35}   is proved.
\end{proof}
\begin{remark}\label{nicola1}
The boundedness of the random value $\xi_1$ is not essential. For applications, the case, as $a_i(\omega_1,\ldots, \omega_{i-1})=0, \ i \geq n+1,$ is very  important (see Section 8).  In this case, Lemma   \ref{1vitanick35} is true as the random value $\eta_1$ is an integrable one.  The random value $\xi_0$ is also   integrable one relative to every measures from the set $M_0$ and it is   ${\cal F}_n$-measurable one.
\end{remark}

Below, we describe completely the regular set of measures in the  case as $\xi_0=\prod\limits_{i=1}^N[1+a_i(\omega_1, \ldots, \omega_{i-1})\eta_i(\omega_i)], \ N< \infty,  $ $0 < a_i(\omega_1, \ldots, \omega_{i-1}) \leq 1,\ i=\overline{1, N},$ and the random value $\xi_1$ is an integrable one relative to the measure $P_1.$ For this purpose, we introduce the denotations: $ \Omega_1^-=\{ \omega_1 \in  \Omega_1^0, \ \eta_1( \omega_1) \leq 0\},$ \ $ \Omega_1^+=\{ \omega_1 \in  \Omega_1^0, \ \eta_1( \omega_1) > 0\},$ $ P_1^- $ is a contraction of the measure $P_1$ on the $\sigma$-algebra ${\cal F}_1^-,$ \ $ P_1^+ $ is a contraction of the measure $P_1$ on the $\sigma$-algebra ${\cal F}_1^+, $  ${\cal F}_1^-=\Omega_1^-\cap {\cal F}_1^0$
${\cal F}_1^+=\Omega_1^+\cap {\cal F}_1^0.$

 Denote $U_1=\Omega_1^- \times \Omega_1^+$ and introduce  the measure $\mu_1= P_1^-\times P_1^+$ on the $\sigma$-algebra $G_1= {\cal F}_1^-\times {\cal F}_1^+.$
Let us introduce the measurable space  $\{{\cal V}, {\cal L}, \mu\},$ where
${\cal V}=\prod\limits_{i=1}^N U_i, \ U_i=U_1, \ i=\overline{1,N},$ is a direct product of the spaces $U_i=\Omega_i^- \times \Omega_i^+, \ \Omega_i^-= \Omega_1^-, \  \Omega_i^+= \Omega_1^+,$  ${\cal L}=\prod\limits_{i=1}^N G_i$ is a direct product of the $\sigma$-algebras $G_i=G_1,  \ i=\overline{1,N}.$
At last, let $\mu=\prod\limits_{i=1}^N\mu_i$ be a direct product of the  measures $\mu_i=\mu_1, \ i=\overline{1,N},$ and let
$\nu_{v}=\prod\limits_{i=1}^N\nu_{\omega_i^1, \omega_i^2}, \ v=\{(\omega_1^1, \omega_1^2), \ldots, (\omega_N^1, \omega_N^2)\}, $ be a direct product of the measures $\nu_{\omega_i^1, \omega_i^2}, \ i=\overline{1, N},$ which is a countable additive function on the  $\sigma$-algebra ${\cal F}_N$ for every $v \in  {\cal L}, $ where
\begin{eqnarray}\label{vitusjanick1}
\nu_{\omega_i^1, \omega_i^2}(A_i)=\chi_{A_i}(\omega_i^1) \frac{\eta_i^+(\omega_i^2)}{\eta_i^-(\omega_i^1) +\eta_i^+(\omega_i^2)}+\chi_{A_i}(\omega_i^2) \frac{\eta_i^-(\omega_i^1)}{\eta_i^-(\omega_i^1) +\eta_i^+(\omega_i^2)}
\end{eqnarray}
for   $\omega_i^1 \in \Omega_i^-,$  $ \omega_i^2 \in \Omega_i^+,$ $A_i\in {\cal F}_i^0.$ 

In the next Theorem \ref{vitusjanick2}, we 
assume that the random value $\eta_1(\omega_1) $ is an integrable one.
\begin{thm}\label{vitusjanick2}
On the measurable space  $\{\Omega, {\cal F}\}$ with the filtration ${\cal F}_n$ on it,   every measure $Q$ of the regular  set of measures $M$ for the random value $\xi_0=\prod\limits_{i=1}^N[1+a_i(\omega_1, \ldots, \omega_{i-1})\eta_i(\omega_i)], \ N< \infty,  $ $0 < a_i(\omega_1, \ldots, \omega_{i-1}) \leq 1,\ i=\overline{1, N},$ has  the representation
\begin{eqnarray}\label{vitusjanick3}
Q(A)=\int\limits_{{\cal V}}\alpha(v)\nu_{v}(A)d\mu(v),
\end{eqnarray}
where the random value $\alpha(v)$  satisfies the conditions
\begin{eqnarray}\label{vitusjanick4}
\mu(\{v \in  {\cal V}, \ \alpha(v)>0\})=[P_1(\Omega_1^-)P_1(\Omega_1^+)]^N,   \end{eqnarray}
\begin{eqnarray}\label{vitusjanick5}
\int\limits_{{\cal V}}\alpha(v)\prod\limits_{i=1}^N\frac{\eta_i^-(\omega_i^1)\eta_i^+(\omega_i^2)}{\eta_i^-(\omega_i^1)+\eta_i^+(\omega_i^2)}d\mu(v)< \infty,
\end{eqnarray}
\begin{eqnarray}\label{vitusjanick6}
\int\limits_{{\cal V}}\alpha(v)d\mu(v)=1.
\end{eqnarray}
\end{thm}
\begin{proof}
To prove Theorem, it needs to prove that  the  countable additive measure  $\nu_{v}(A)$ at every fixed $v \in {\cal V}$ is a measurable map from the measurable space
 $\{{\cal V}, {\cal L}\}$ into the measurable space $\{[0,1], B([0,1])\}$ for every fixed $A \in {\cal F}_N.$  For $A=\prod\limits_{i=1}^N A_i, \ A_i \in {\cal F}_i^0,$
 $\nu_{v}(A)$ is a  measurable map from the measurable space
 $\{{\cal V}, {\cal L}\}$ into the measurable space $\{[0,1], B([0,1])\}.$
The family of sets of the kind $\bigcup\limits_{i \in I} E_i, $ $E_i=\prod\limits_{s=1}^N A_s^i, \ A_s^i \in {\cal F}_s^0,$ where $E_i \cap E_j=\emptyset,$ the set $I$ is an arbitrary finite set, forms the algebra of the sets that we denote by $U_0.$ From the countable additivity of $\nu_{v}(A),$
$ \nu_{v}(\bigcup\limits_{i \in I} E_i)=\sum\limits_{i \in I} \nu_{v}(E_i)$ is a measurable map from the measurable space
 $\{{\cal V}, {\cal L}\}$ into the measurable space $\{[0,1], B([0,1])\}.$
Let $T$ be a class  of the sets  from the minimal $\sigma$-algebra $\Sigma$ generated by $U_0$ for every  subset $E$ of that $\nu_{v}(E)$ is a measurable map
from the measurable space
 $\{{\cal V}, {\cal L}\}$ into the measurable space $\{[0,1], B([0,1])\}.$
Let us prove that $T$ is a monotonic class. Suppose that  $E_i \subset E_{i+1}, \  i=\overline{1, \infty}, \  E_i \in T.$ Then,  $\nu_{v}(E_i) \leq \nu_{v}(E_{i+1}).$ From this, it follows that $\lim\limits_{i \to \infty} \nu_{v}(E_i)$ is a  measurable map from the measurable space
 $\{{\cal V}, {\cal L}\}$ into the measurable space $\{[0,1], B([0,1])\}.$
But,   $\nu_{v}(E_{i+1} \setminus E_{i})=\nu_{v}(E_{i+1}) - \nu_{v}(E_i) $ is a measurable map from  $\{{\cal V}, {\cal L}\}$ into  $\{[0,1], B([0,1])\}.$
From this equality, it follows that the set $E_{i+1} \setminus E_i$ belongs to the class $T.$
Since $\bigcup\limits_{i=1}^\infty E_i=E_1\cup \bigcup\limits_{i=1}^\infty [ E_{i+1}\setminus E_i],$ we have
$$ \lim\limits_{ n \to \infty}\nu_{v}(E_n)= \nu_{v}(E_1)+  \lim\limits_{ n \to \infty}\sum\limits_{i=1}^n \nu_{v}(E_{i+1} \setminus E_i)=$$
\begin{eqnarray}\label{vitusjanick7}
\nu_{v}(E_1)+ \sum\limits_{i=1}^\infty \nu_{v}(E_{i+1} \setminus E_i)= \nu_{v}(E_1\cup \bigcup\limits_{i=1}^\infty [ E_{i+1}\setminus E_i]) =\nu_{v}(\bigcup\limits_{i=1}^\infty E_i). 
\end{eqnarray}
The equalities (\ref{vitusjanick7}) mean that $\bigcup\limits_{i=1}^\infty E_i$
belongs to $T,$ since $\nu_{v}(\bigcup\limits_{i=1}^\infty E_i)$ is a measurable map
of  $\{{\cal V}, {\cal L}\}$ into  $\{[0,1], B([0,1])\}.$
Suppose that $E_i \supset E_{i+1}, \ E_i \in T,  i=\overline{1,\infty}.$
Then, this case is reduced to the previous one by the note that the sequence $\bar E_i=\prod\limits_{i=1}^N\Omega_i^0 \setminus E_i, \  i=\overline{1,\infty}$ is monotonically increasing. From this, it follows that $\bar E = \bigcup\limits_{i=1}^\infty\bar  E_i \in T.$ Therefore, $\bigcap\limits_{i=1}^\infty E_i=\prod\limits_{i=1}^N\Omega_i^0 \setminus \bigcup\limits_{i=1}^\infty\bar E_i \in T.$ Thus, $T$ is a monotone class. But, $U_0 \subset T.$ Hence, T contains the minimal monotone class generated by the algebra $U_0,$ that is, $m(U_0)=\Sigma,$
therefore, $\Sigma \subset T.$ Thus, $ \nu_{v}(E) $ is a measurable map of $\{{\cal V}, {\cal L}\}$ into  $\{[0,1], B([0,1])\}$ for $A \in \Sigma.$
The fact that the random value $\alpha(v)$ satisfies the conditions
(\ref{vitusjanick4}) - (\ref{vitusjanick6}) means that  $Q,$ given by the formula (\ref{vitusjanick3}), is a countable additive function of sets  and $ E^Q\xi_0<\infty.$
Moreover, $E^Q\xi_0=1.$  It is evident that
$E^Q\{\xi_0| {\cal F}_n\}=\prod\limits_{i=1}^n [1+a_i(\omega_1, \ldots, \omega_{i-1})\eta_i(\omega_i)], \ Q \in M.$ Due to Lemma \ref{vitanick1}, this proves that the set $M$ is a regular set of measure.
Theorem \ref{vitusjanick2} is proved.

\end{proof}

\begin{remark}\label{vitakolja1}
The representation (\ref{vitusjanick3}) for the regular set of measures $M$
means that  $M$ is a convex set of equivalent measures. Since the random value $\alpha(v)$ runs all bounded random value, satisfying the conditions (\ref{vitusjanick4} - \ref{vitusjanick6}), it is easy to show that the set of measures
$\nu_v(A), \  v \in {\cal V}, \ A \in {\cal F}_N,$ is the set of extreme points for the set  $M.$ Moreover, since in the representation (\ref{vitusjanick3}) for the regular set of measures $M$ $\alpha(v)$ runs all  bounded random values,  satisfying the conditions (\ref{vitusjanick4} - \ref{vitusjanick6}), then  $M$ is a complete set of measures. 
\end{remark}

\begin{thm}\label{4vitanick3}
On the probability space  $\{\Omega, {\cal F}, P\} $ with the filtration ${\cal F}_n$ on it,  let $M$ be a complete  set of  measures.   If 
every $\sigma$-algebra ${\cal F}_n, \ n=\overline{1,\infty},$  has an  exhaustive decomposition, then 
 the closure of the set of points $E^Q f_n(\omega), \ Q \in M_n,$ in metrics $\rho(x,y)=|x-y|, \ x, y \in R^1,$  contains the set of points
$$ f_n(\omega_1)\frac{[m_n-m_{n-1}]^+(\omega_2)}{[m_n-m_{n-1}]^-(\omega_1)+[m_n-m_{n-1}]^+(\omega_2)}+$$
\begin{eqnarray}\label{4vitanick2}
f_n(\omega_2)\frac{[m_n-m_{n-1}]^-(\omega_1)}{[m_n-m_{n-1}]^-(\omega_1)+[m_n-m_{n-1}]^+(\omega_2)},
\end{eqnarray}
$$  (\omega_1,\omega_2) \in  \Omega_n^-\times  \Omega_n^+, \quad n=\overline{1, \infty},$$
 for every integrable relative to every measure $Q \in M_n$ the finite valued ${\cal F}_n$-measurable  random value $f_n(\omega),$ 
where $ \Omega_n^-=\{\omega_1 \in \Omega, \  [m_n-m_{n-1}] (\omega_1)\leq 0\}, $ 
$ \Omega_n^+=\{\omega_2 \in \Omega, \  [m_n-m_{n-1}](\omega_2) > 0\}. $ 
\end{thm}
\begin{proof}
 Since
$$E^{Q} f_n(\omega)=\int\limits_{\Omega_n^-\times\Omega_n^+}f_n(\omega_1)\frac{\alpha_n(\omega_1,\omega_2)[m_n-m_{n-1}]^+(\omega_2)}{[m_n-m_{n-1}]^-(\omega_1)+[m_n-m_{n-1}]^+(\omega_2)}d\mu_n(\omega_1,\omega_2)+$$
\begin{eqnarray}\label{77vitanick2}
\int\limits_{\Omega_n^-\times\Omega_n^+}f_n(\omega_2)\frac{\alpha_n(\omega_1,\omega_2)[m_n-m_{n-1}]^-(\omega_1)}{[m_n-m_{n-1}]^-(\omega_1)+[m_n-m_{n-1}]^+(\omega_2)}d\mu_n(\omega_1,\omega_2).
\end{eqnarray}
Then, all arguments, used in the proof  of Lemma \ref{2vitanick1}, can be applied for the proof of Theorem \ref{4vitanick3},  since $E^{P^n}|m_n-m_{n-1}|<\infty.$ Theorem \ref{4vitanick3} is proved.
\end{proof}
\begin{thm}\label{4vitanick4}
On the probability space  $\{\Omega, {\cal F}, P\} $ with the filtration ${\cal F}_n$ on it, let   $M$ be a complete set of  measures and let  every $ \sigma$-algebra ${\cal F}_n, \ n=\overline{1, \infty}, $ have an  exhaustive decomposition.   Suppose that $f_n(\omega)$ is a nonnegative integrable    ${\cal F}_n$-measurable random value, satisfying the condition $E^{Q^n} f_n(\omega) \leq 1, Q^n \in M_n.$ Then, there exists a constant $\alpha_n,$ depending on $f_n(\omega),$ such that 
\begin{eqnarray}\label{4vitanick5}
f_n(\omega) \leq 1+ \alpha_n [m_n-m_{n-1}](\omega), \quad  \omega \in \Omega.
\end{eqnarray}
\end{thm}

\begin{proof}
Due to the completeness of the set of measures $M,$ let us denote  a local  regular martingale by 
$\{m_n, {\cal F}_n\}_{n=0}^\infty, \ m_n=E^Q\{\xi_0|{\cal F}_n\}, Q \in M, \ \xi_0 \in A_0, \ \xi_0 \neq 1.$
 From the completeness  of the set of measures $M,$ 
we obtain the inequality
$$ f_n(\omega_1)\frac{[m_n-m_{n-1}]^+(\omega_2)}{[m_n-m_{n-1}]^-(\omega_1)+[m_n-m_{n-1}]^+(\omega_2)}+$$
\begin{eqnarray}\label{4vitanick6}
f_n(\omega_2)\frac{[m_n-m_{n-1}]^-(\omega_1)}{[m_n-m_{n-1}]^-(\omega_1)+[m_n-m_{n-1}]^+(\omega_2)}\leq 1,
\end{eqnarray}
$$  (\omega_1,\omega_2) \in  \Omega_n^-\times  \Omega_n^+,$$
where $ \Omega_n^-=\{\omega_1 \in \Omega, \  [m_n-m_{n-1}](\omega_1) \leq 0\}, $ 
$ \Omega_n^+=\{\omega_2 \in \Omega, \  [m_n-m_{n-1}](\omega_2 ) > 0\}. $ 

Let us denote $\xi_n(\omega)= [m_n-m_{n-1}](\omega).$ Then, the formula (\ref{4vitanick6}) is written in the form
$$ f_n(\omega_1) \frac{\xi_n^+(\omega_2)}{\xi_n^-(\omega_1) +\xi_n^+(\omega_2)} +$$
\begin{eqnarray}\label{3vitamyk15}
 \frac{\xi_n^-(\omega_1)}{\xi_n^-(\omega_1)+
\xi_n^+(\omega_2)} f_n(\omega_2)\leq 1, \quad \omega_1 \in \Omega_n^-, \quad \omega_2 \in \Omega_n^+.
\end{eqnarray} 
From  the inequalities  (\ref{3vitamyk15}),  we obtain the inequalities
\begin{eqnarray}\label{3vitamyk16}
f_n (\omega_2)\leq 1+\frac{1- f_n(\omega_1)}{\xi_n^-(\omega_1)}\xi_n^+(\omega_2), \end{eqnarray} 
\begin{eqnarray}
 \xi_n^-(\omega_1) >0, \quad  \xi_n^+(\omega_2) >0, 
\quad  \omega_1 \in \Omega_n^-, \quad \omega_2 \in \Omega_n^+.
\end{eqnarray}
Two cases are possible: a) for all  $ \omega_1 \in \Omega_n^-,$ $ f_n(\omega_1) \leq 1; $ b) there exists $ \omega_1 \in \Omega_n^-$ such that $  f_n (\omega_1)> 1.$
First, let us consider the case a).

Since the inequalities  (\ref{3vitamyk16}) are valid for every value  $\frac{1- f_n(\omega_1)}{ \xi_n^-(\omega_1)},$ as $  \xi_n^-(\omega_1) >0, $ and $ f_n(\omega_1) \leq 1, \omega_1 \in \Omega_n^- ,$  then, if to denote
\begin{eqnarray}
\alpha_n =\inf_{\{\omega_1,  \xi_n^-(\omega_1) > 0\}}\frac{1- f_n(\omega_1)}{ \xi_n^-(\omega_1)},
\end{eqnarray}
 we have $ 0 \leq \alpha_n < \infty$  and
\begin{eqnarray}\label{3vitamyk17}
f_n (\omega_2)\leq 1+\alpha_n  \xi_n^+(\omega_2),  \quad    \xi_n^+ (\omega_2)>0,  \quad \omega_2 \in \Omega_n^+.
\end{eqnarray} 
From the definition of  $\alpha_n,$ we obtain the inequalities 
\begin{eqnarray}
f_n(\omega_1)  \leq 1-\alpha_n \xi_n^-(\omega_1),  \quad   \xi_n^-(\omega_1) >0, \quad  \omega_1 \in \Omega_n^-.
\end{eqnarray} 
Now, if $\xi_n^-(\omega_1)=0$ for some $ \omega_1\in \Omega_n^-, $ then in this case $f_n(\omega_1) \leq 1.$ All these inequalities give the inequalities
\begin{eqnarray}\label{3vitamyk18}
f_n(\omega) \leq 1+\alpha_n \xi_n(\omega),  \quad \omega \in  \Omega_n^-\cup  \Omega_n^+.
\end{eqnarray} 
Consider the case b). From the inequality (\ref{3vitamyk16}), we obtain  the inequalities
\begin{eqnarray}\label{3vitamyk19}
f_n(\omega_2) \leq 1-\frac{1- f_n(\omega_1)}{-\xi_n^-(\omega_1)}\xi_n^+(\omega_2),
\end{eqnarray}
\begin{eqnarray}
  \xi_n^-(\omega_1) > 0, \quad  \xi_n^+(\omega_2) >0, \quad  \omega_1 \in  \Omega_n^-, \quad \omega_2 \in  \Omega_n^+.
\end{eqnarray}
 The  inequalities (\ref{3vitamyk19}) give  the inequalities
\begin{eqnarray}\label{3vitamyk20}
\frac{1- f_n(\omega_1)}{-\xi_n^-(\omega_1)} \leq \inf_{\{\omega_2, \ \xi_n^+(\omega_2) >0\}} \frac{1}{\xi_n^+(\omega_2)}< \infty, \quad  \xi_n^-(\omega_1) >0, \quad  \omega_1 \in \Omega_n^-.
\end{eqnarray}
Let us define $\alpha_n =\sup\limits_{\{ \omega_1, \ \xi_n^-(\omega_1)>0 \} }\frac{1- f_n(\omega_1)}{-\xi_n^-(\omega_1)}< \infty.$
Then, from (\ref{3vitamyk19}) we obtain  the inequalities
\begin{eqnarray}\label{3vitamyk21}
f_n(\omega_2)   \leq 1- \alpha_n \xi_n^+(\omega_2), \quad  \  \xi_n^+(\omega_2) >0, \quad \omega_2 \in \Omega_n^+.
\end{eqnarray}
From the definition of  $\alpha_n, $ we have the inequalities
\begin{eqnarray}\label{3vitamyk22}
f_n(\omega_1) \leq 1+ \alpha_n \xi_n^-(\omega_1), \quad  \  \xi_n^-(\omega_1) >0, \quad  \omega_1 \in \Omega_n^-.
\end{eqnarray}
The inequalities (\ref{3vitamyk21}), (\ref{3vitamyk22}) give  the inequalities
\begin{eqnarray}\label{3vitamyk23}
f_n(\omega) \leq 1- \alpha_n \xi_n(\omega),  \quad  \omega \in   \Omega_n^-\cup   \Omega_n^+.
\end{eqnarray}
Theorem  \ref{4vitanick4}  is proved, since the set $  \Omega_n^-\cup   \Omega_n^+$ has the probability one.
\end{proof}

\begin{thm}\label{4vitanick7}
On the probability space  $\{\Omega, {\cal F}, P\} $ with the filtration ${\cal F}_n$ on it, let  $M$ be a complete  set of measures and let every $ \sigma$-algebra ${\cal F}_n, \ n=\overline{1, \infty}, $ have an  exhaustive decomposition.  Then, every nonnegative super-martingale $\{f_n, {\cal F}_n\}_{n=0}^\infty$ is a local regular one.
\end{thm}
\begin{proof} Without loss of generality, we assume that $f_n \geq  d_0>0.$ From the last fact, we obtain
\begin{eqnarray}\label{4vitanick8}
E^{Q^n}\frac{f_n}{f_{n-1}} \leq 1, \quad Q^n \in M_n, \quad n=\overline{1, \infty}.
\end{eqnarray}
The inequalities (\ref{4vitanick8}) and Theorems \ref{nick1}, \ref{4vitanick4}  prove Theorem \ref{4vitanick7}.
\end{proof}

\begin{thm}\label{400vitanick7}
On the probability space  $\{\Omega, {\cal F}\} $ with the filtration ${\cal F}_n$ on it, let  $M$ be a  complete set of  measures and let every $ \sigma$-algebra ${\cal F}_n, \ n=\overline{1, \infty}, $ have an   exhaustive decomposition.  Then, every bounded from below super-martingale $\{f_n, {\cal F}_n\}_{n=0}^\infty$ is a local regular one.
\end{thm}

\begin{proof} Since the super-martingale  $\{f_n, {\cal F}_n\}_{n=0}^\infty$  is bounded from below, then there exists a real number $C_0$ such that  $f_n +  C_0>0.$ If to consider the super-martingale
$\{f_n+C_0, {\cal F}_n\}_{n=0}^\infty,$ then all conditions of Theorem  \ref{4vitanick7} are true. Theorem \ref{400vitanick7} is proved.
\end{proof}

\section{Description of local regular super-martingales relative to a regular set of  measures.}

In this section, we give the description of local regular super-martingales.

\begin{thm}\label{mmars1} On the measurable space $\{\Omega, {\cal F}\}$ with the filtration ${\cal F}_n$ on it, let $M$ be a regular  set of  
measures. 
If $\{f_m, {\cal F}_m\}_{m=0}^\infty$ is an adapted random process,  satisfying the conditions
\begin{eqnarray}\label{mmars2}
f_m \leq f_{m-1}, \quad  E^P\xi|f_m| <\infty, \quad P \in M \quad m=\overline{1, \infty}, \quad  \xi \in A_0,
\end{eqnarray}
then the random process
\begin{eqnarray}\label{mmars3}
 \{ f_mE^P\{\xi|{\cal F}_m\}, {\cal F}_m\}_{m=0}^\infty, \quad  P \in M,
\end{eqnarray}
is a local regular super-martingale relative to the regular set of   measures  $M.$
\end{thm}
\begin{proof} Due to Theorem  \ref{tatnick8}, the random process 
$\{ E^P\{\xi|{\cal F}_m\}, {\cal F}_m\}_{m=0}^\infty$ is a martingale relative to the regular  set of   measures  $M.$ Therefore,
\begin{eqnarray*}
f_{m-1}E^P\{\xi|{\cal F}_{m-1}\} - E^P\{ f_m E^P\{\xi|{\cal F}_m\}|{\cal F}_{m-1}\}=
\end{eqnarray*}
\begin{eqnarray}\label{mmars4}
E^P\{ (f_{m-1} - f_m) E^P\{\xi|{\cal F}_m\}|{\cal F}_{m-1}\}, \quad m=\overline{1, \infty}.
\end{eqnarray}
So, if to put  $\bar g_m^0= (f_{m-1} - f_m) E^P\{\xi|{\cal F}_m\}, \ m=\overline{1, \infty}, $ then $\bar g_m^0 \geq 0,$  it is  ${\cal F}_m$-measurable and
$E^P\bar g_m^0 \leq E^P\xi(|f_{m-1}|+|f_m|)< \infty.$ Due to Theorem \ref{reww1}, we obtain the proof of Theorem \ref{mmars1}.
\end{proof}

\begin{cor}\label{hg1} If $f_m=\alpha, \ m=\overline{1, \infty}, \ \alpha \in R^1,$ $\xi \in A_0,$ then 
$\{\alpha E^P\{\xi |{\cal F}_m\},  {\cal F}_m \}_{m=0}^\infty$ is a local regular martingale. Assume that  $\xi =1,$ then $\{f_m, {\cal F}_m\}_{m=0}^\infty$ is a local regular super-martingale relative to the regular  set of   measures  $M.$ 
\end{cor}

Denote  $F_0$ the set of adapted processes
\begin{eqnarray}\label{mmars5}
F_0=\{ f=\{f_m\}_{m=0}^\infty,  \   P(|f_m| <\infty) =1, \ P \in M, \ f_m \leq f_{m-1}\}.
\end{eqnarray}
For every $\xi \in A_0,$ let us introduce the set of adapted processes
$$ L_{\xi}=$$
\begin{eqnarray}\label{mmars6}
\{\bar f=\{f_mE^P\{\xi|{\cal F}_m\}\}_{m=0}^\infty, \  \{f_m\}_{m=0}^\infty \in F_0, \   E^P\xi|f_m| <\infty, \ P \in M\},
\end{eqnarray}
and 
\begin{eqnarray}\label{mmars7}
V=\bigcup\limits_{\xi \in A_0}L_{\xi}.
\end{eqnarray}

\begin{cor}\label{fdr1} Every  random  process from the set $K,$ where
\begin{eqnarray}\label{mmars88}
K=\left \{ \sum\limits_{i=1}^mC_i \bar f_i, \ \bar f_i \in V, \  C_i \geq 0, \ i=\overline{1, m}, \ m=\overline{1, \infty}\right\}, 
\end{eqnarray}
 is a local regular super-martingale relative to the regular  set of  measures  $M$  on the measurable space $\{\Omega, {\cal F}\}$ with the filtration ${\cal F}_n$ on it. 
\end{cor}
\begin{proof} The proof is evident.
\end{proof}

\begin{thm}\label{mmars9} On the measurable space $\{\Omega, {\cal F}\}$ with the filtration ${\cal F}_n$ on it, let $M$ be a regular  set of 
measures. 
 Suppose that  $\{f_m, {\cal F}_m\}_{m=0}^\infty$ is a nonnegative uniformly integrable super-martingale relative to the set of  measures  $M,$ then 
the necessary and sufficient conditions for it  to be a local regular one is belonging it to the set $K.$
\end{thm}
\begin{proof}
The necessity.  It is evident  that if  $\{f_m, {\cal F}_m\}_{m=0}^\infty$ belongs to $K,$ then it is a local regular super-martingale.

The sufficiency. Suppose that  $\{f_m,{\cal F}_m\}_{m=0}^\infty$ is a nonnegative uniformly integrable local regular super-martingale. Then, there exists a nonnegative adapted process $\{\bar g_m^0\}_{m=1}^ \infty,  \ E^P\bar g_m^0< \infty, \ m=\overline{1, \infty}, $  and a martingale  $\{M_m, {\cal F}_m\}_{m=0}^ \infty,$
such that 
\begin{eqnarray}\label{mmars8}
f_m=M_m - \sum\limits_{i=1}^m\bar g_i^0, \quad  m=\overline{0, \infty}. 
\end{eqnarray}
Then, $M_m \geq 0, \ m=\overline{0, \infty}, \ E^P M_m <\infty, \ P\in M.$
Since $0< E^PM_m=f_0< \infty,$ we have $E^P\sum\limits_{i=1}^m\bar g_i^0< f_0.$ Let us put $g_{\infty}=\lim\limits_{m \to \infty}\sum\limits_{i=1}^m\bar g_i^0.$
Using the uniform integrability of $\{f_m, {\cal F}_m\}_{m=0}^\infty,$ we can pass to the limit in the equality
\begin{eqnarray}\label{apm23}
 E^P(f_m +\sum\limits_{i=1}^m\bar g_i^0)=f_0, \quad P \in M,
\end{eqnarray}
as $m \to \infty$.
Passing to the limit in the last equality, as $m \to \infty,$  we obtain
\begin{eqnarray}\label{apm23}
E^P(f_\infty +g_{\infty})=f_0, \quad P \in M.
\end{eqnarray}
Introduce into consideration a  random value $\xi=\frac{f_\infty +g_{\infty}}{f_0}.$
Then, $E^P\xi=1, \ P \in M.$   From here, we obtain that  $\xi \in A_0$ and 
\begin{eqnarray}\label{apm24}
M_m=f_0E^P\{\xi|{\cal F}_m\}, \ m=\overline{0, \infty}.
\end{eqnarray}
 Let us put $\bar f_m^2=-\sum\limits_{i=1}^m\bar g_i^0. $ It is easy to see that the adapted random process $\bar f_2=\{\bar f_m^2, {\cal F}_m\}_{m=0}^\infty$ belongs to $F_0.$ Therefore,
for the super-martingale $f=\{ f_m,{\cal F}_m\}_{m=0}^\infty$ the representation 
$$f=\bar f_1+ \bar f_2,$$
is valid, where $\bar f_1=\{f_0E^P\{\xi|{\cal F}_m\}, {\cal F}_m\}_{m=0}^ \infty$  belongs to $L_{\xi}$
with  $ \xi = \frac{f_\infty +g_{\infty}}{f_0}$  and $ f_m^1=f_0, \ m=\overline{0,\infty}.$ The same is valid for $\bar f_2$ with $\xi=1.$ This implies that $f$ belongs to the set $K.$  Theorem 
\ref{mmars9} is proved.
\end{proof}
\begin{thm}\label{9mmars9} On the measurable space $\{\Omega, {\cal F}\}$ with the filtration ${\cal F}_n$ on it, let $M$ be a regular  set of
measures. Suppose that  the super-martingale $\{f_m, {\cal F}_m\}_{m=0}^\infty$   relative to the set of  measures  $M$ satisfy the conditions
\begin{eqnarray}\label{99mmars88}
|f_m|\leq C \xi_0, \quad m=\overline{1, \infty}, \quad \xi_0 \in A_0, \quad 0<C<\infty,
\end{eqnarray}
 then the necessary and sufficient conditions for it  to be a local regular one is belonging it to the set $K.$
\end{thm}
\begin{proof} The necessity is evident.\\
 The sufficiency.  
 Suppose that  $\{f_m,{\cal F}_m\}_{m=0}^\infty$ is a local regular super-martingale. Then, there exists a nonnegative adapted random process $\{\bar g_m^0\}_{m=1}^ \infty,  \ E^P\bar g_m^0< \infty, \ m=\overline{1, \infty}, $  and a martingale  $\{M_m\}_{m=0}^ \infty, \ E^P| M_m| <\infty, \  m=\overline{1, \infty}, \  P\in M,$
such that 
\begin{eqnarray}\label{9mmars8}
f_m=M_m - \sum\limits_{i=1}^m\bar g_i^0, \quad  m=\overline{0, \infty}. 
\end{eqnarray}
The inequalities $f_m +C\xi_0 \geq 0, \  m=\overline{1, \infty}, $ give the inequalities
\begin{eqnarray}\label{8mmars8}
 f_m+C E^P\{\xi_0|{\cal F}_m\}\geq 0, \quad  m=\overline{0, \infty}.
\end{eqnarray}
From the inequalities (\ref{99mmars88}), it follows that the super-martingale $\{f_m,{\cal F}_m\}_{m=0}^\infty$ is a uniformly integrable one  relative  to the regular set of   measures $M$.  The martingale  $\{E^P\{\xi_0|{\cal F}_m\}, {\cal F}_m\}_{m=0}^\infty$ relative to the   regular set of measures $M$ is also uniformly integrable one.

Then, $M_m+C E^P\{\xi_0|{\cal F}_m\} \geq 0, \ m=\overline{0, \infty}. $
Since $0< E^P[M_m+C E^P\{\xi_0|{\cal F}_m\}]=f_0+C< \infty,$ we have $E^P\sum\limits_{i=1}^m\bar g_i^0< f_0+C.$ Let us put $g_{\infty}=\lim\limits_{m \to \infty}\sum\limits_{i=1}^m\bar g_i^0.$
Using the uniform integrability of $f_m$ and $\sum\limits_{i=1}^m\bar g_i^0,$ we can pass to the limit in the equality
\begin{eqnarray}\label{apm28}
 E^P(f_m +\sum\limits_{i=1}^m\bar g_i^0)=f_0, \quad P \in M,
\end{eqnarray}
as $m \to \infty$.
Passing to the limit in the last equality, as $m \to \infty,$  we obtain
\begin{eqnarray}\label{apm24}
E^P(f_\infty +g_{\infty})=f_0, \quad P \in M.
\end{eqnarray}
Introduce into consideration a  random value $\xi_1=\frac{f_\infty+C \xi_0+g_{\infty}}{f_0+C}\geq 0.$
Then, $E^P\xi_1=1, \ P \in M.$ From here, we obtain that  $\xi_1 \in A_0$ and for the super-martingale  $f=\{ f_m,{\cal F}_m\}_{m=0}^\infty$   the representation 
\begin{eqnarray}\label{apm25}
 f_m=f_m^0E^P\{\xi_0|{\cal F}_m\}+f_m^1 E^P\{\xi_1|{\cal F}_m\}+f_m^2E^P\{\xi_2|{\cal F}_m\} , \ m=\overline{0, \infty},
\end{eqnarray}
is valid, where $f_m^0=-C, \ f_m^1=f_0+C, \ f_m^2=-\sum\limits_{i=1}^m\bar g_i^0, \ m=\overline{0, \infty}, \ \xi_2=1.$
From  the last representation, it follows that  the super-martingale $f=\{ f_m,{\cal F}_m\}_{m=0}^\infty$ belongs to the set $K.$ Theorem 
\ref{9mmars9} is proved.
\end{proof}

\begin{cor}\label{mars16}  Let $f_N, \  N< \infty,$ be a ${\cal F}_N$-measurable integrable random value,  $\sup\limits_{P \in M} E^P|f_N| < \infty,$ and let there exist $\alpha_0 \in R^1$ such that
$$ -\alpha_0  M_N+ f_N \leq 0, \quad \omega \in \Omega, $$
where $\{ M_m, {\cal F}_m\}_{m=0}^\infty=\{E^P\{\xi|{\cal F}_m\}, {\cal F}_m\}_{m=0}^\infty, \ \xi \in A_0. $ 
Then, a  super-martingale $\{ f_m^0+ \bar f_m\}_{m=0}^\infty$ is a local regular one relative to  the regular set of  measures $M,$ where
\begin{eqnarray}\label{apm26}
f_m^0=\alpha_0 M_m,  
\end{eqnarray}
\begin{eqnarray}\label{apm27}
\bar f_m=
\left\{\begin{array}{l l} 0, & m<N, \\
f_N - \alpha_0 M_N, & m \geq N.  
 \end{array} \right. 
\end{eqnarray}
\end{cor}
\begin{proof}  It is evident that $\bar f_{m-1} -\bar f_m \geq 0, \  m=\overline{0,\infty}.$
Therefore, the super-martingale
\begin{eqnarray}\label{apm27}
f_m^0+ \bar f_m=
\left\{\begin{array}{l l} \alpha_0 M_m, & m<N, \\
f_N , & m= N,  \\
f_N - \alpha_0 M_N+\alpha_0 M_m, & m>N
 \end{array} \right. 
\end{eqnarray}
is a local regular one relative to  the regular set of measures $M.$
 Corollary \ref{mars16} is proved.
\end{proof}

\section{Optional decomposition for  super-martingales relative to a complete set of  measures.}

 In this section, we  prove that the bounded super-martingales are local regular ones with respect to the complete set of  measures.
\subsection{Measurable space  with a finite decomposition.}
In this and the  next subsections, we reformulate the results of the paper \cite{GoncharNick}.
Let $\{\Omega, {\cal F}\}$ be a measurable space. We assume that 
the $\sigma$-algebra  ${\cal F}$ is   a certain finite algebra of subsets of the set   $\Omega.$  
We give a new proof of the optional decomposition for  super-martingales relative to the complete set of measures. This proof  does not use topological arguments as  in  \cite{WalterSchacher}.
 Let  ${\cal F}_n \subset {\cal F}_{n+1} \subset {\cal F} $ be an increasing set of  algebras, where   ${\cal F}_0 =\{\emptyset, \Omega\}, $
  ${\cal F}_N = {\cal F}. $ Denote  $M$ the  complete set of  measures on the measurable space $\{\Omega, {\cal F}\}.$   It is evident that every algebra ${\cal F}_n$ is generated by sets $A_i^n, \   i=\overline{1, N_n}, A_i^n\cap A_j^n=\emptyset, \ i \neq j, \ N_n<\infty, \ \bigcup\limits_{i=1}^{N_n}A_i^n=\Omega, \ n=\overline{1,N}.$ It is evident that such decompositions are exhastive one.
Let  $m_n=E^P\{\xi_0|{\cal F}_n\}, \ P \in M,  \   n=\overline{1, N}, \ \xi_0 \in A_0.$ Then, for $m_n$ the representation 
\begin{eqnarray}\label{1myk1}
m_n=\sum\limits_{i=1}^{N_n}m_i^n\chi_{A_i^n}(\omega), \quad n=\overline{1,N},
\end{eqnarray} 
is valid.

\begin{lemma}\label{1myk10} Let $M$ be  a  complete set of   measures on the measurable space  $\{\Omega, {\cal F}\}$ with the filtration ${\cal F}_n$ on it. Then, for every non negative bounded  ${\cal F}_n$-measurable random value $\xi_n=\sum\limits_{i=1}^{N_n} C_i^n \chi_{A_i^n}$ there exists a real number $\alpha_n$ such that
\begin{eqnarray}\label{1myk11} 
f_n(\omega)=\frac{ \sum\limits_{i=1}^{N_n} C_i^n \chi_{A_i^n} }{\sup\limits_{P \in M_n}\sum\limits_{i=1}^{N_n} C_i^n P(A_i^n)} \leq 1+\alpha_n (m_n - m_{n-1}), \quad 
n=\overline{1,N}.
\end{eqnarray}
\end{lemma}
 \begin{proof} 
The random value $f_n(\omega)$ satisfy all conditions of Theorems  \ref{4vitanick3}, \ref{4vitanick4}. This proves Lemma \ref{1myk10}.
\end{proof}

\begin{thm}\label{1myk19}
  Let $M$ be  a   complete set of    measures on the measurable space  $\{\Omega, {\cal F}\}$ with the filtration ${\cal F}_n$ on it.
Then,  every non negative super-martingale  $\{f_m, {\cal F}_m\}_{m=0}^N$   relative to the set of  measures $M$  
 is a local regular one.
\end{thm}
\begin{proof} Without loss of generality, we assume that $f_n >a>0, n=\overline{1,N}.$ Then, the random value  $\frac{f_n}{f_{n-1}}$ satisfy conditions of Theorems   \ref{4vitanick4},  \ref{4vitanick7}.  Therefore, all conditions of Theorem \ref{nick1} are satisfied. Theorem \ref{1myk19} is proved.
\end{proof}

\begin{thm}\label{myktinal4}
  Let $M$ be a complete set of   measures on the measurable space  $\{\Omega, {\cal F}\}$ with the filtration ${\cal F}_n$ on it.
Then,  every bounded super-martingale  $\{f_m, {\cal F}_m\}_{m=0}^N$ relative to the    set of  measures $M$ is a local regular one.
\end{thm}
\begin{proof} From the boundedness of super-martingale  $\{f_m, {\cal F}_m\}_{m=0}^N,$ there exists a constant $C_0>0$ such that
$\frac{3C_0}{2}> f_m+C_0>\frac{C_0}{2}, \ \omega \in \Omega, \ m=\overline{0,N}.$ From this, it follows that the super-martingale  $\{f_m+C_0, {\cal F}_m\}_{m=0}^N$ is a nonnegative one
and satisfies the conditions
\begin{eqnarray}\label{apm33}
 \frac{f_n+C_0}{f_{n-1}+C_0}\leq 3, \quad n=\overline{1, N}.
\end{eqnarray}
It implies that the conditions of Theorem \ref{1myk19} are satisfied.
Theorem \ref{myktinal4} is proved.
\end{proof}

\subsection{Measurable space with a countable decomposition.}
In this subsection, we generalize the results of the previous subsection onto the measurable space $\{\Omega, {\cal F}\}$ with the countable decomposition.

Let  ${\cal F}_n \subset {\cal F}_{n+1} \subset {\cal F} $ be   a certain  increasing set  of $\sigma$-algebras, where   ${\cal F}_0 =\{\emptyset, \Omega\}.$
 Suppose that the  $\sigma$-algebra ${\cal F}_n$ is generated by the sets $A_i^n, \   i=\overline{1, \infty}, \  A_i^n\cap A_j^n=\emptyset, \ i \neq j, \ \bigcup\limits_{i=1}^{\infty}A_i^n=\Omega, \ n=\overline{1,\infty}.$ We assume that ${\cal F}=\sigma(\bigvee\limits_{n=0}^\infty{\cal F}_n).$
Denote  $M$ the complete set of   measures on the measurable space $\{\Omega, {\cal F}\}.$   
Introduce into consideration the martingale  $m_n=E^P\{\xi_0|{\cal F}_n\}, \ P \in M, \  n=\overline{1, \infty}, \ \xi_0 \in A_0.$ Then, for $m_n$ the representation 
\begin{eqnarray}\label{2myk1}
m_n=\sum\limits_{i=1}^{\infty}m_i^n\chi_{A_i^n}(\omega), \quad n=\overline{1,\infty},
\end{eqnarray} 
is valid.

\begin{lemma}\label{2myk10}   Let $M$ be  a  complete  set of   measures on the measurable space  $\{\Omega, {\cal F}\}$ with the filtration ${\cal F}_n$ on it. Then, for every non negative bounded ${\cal F}_n$-measurable random value $\xi_n=\sum\limits_{i=1}^{\infty} C_i^n \chi_{A_i^n},$ there exists a real number $\alpha_n$ such that
\begin{eqnarray}\label{2myk11} 
f_n(\omega)=\frac{ \sum\limits_{i=1}^{\infty} C_i^n \chi_{A_i^n} }{\sup\limits_{P \in M_n}\sum\limits_{i=1}^{\infty} C_i^n P(A_i^n)} \leq 1+\alpha_n (m_n - m_{n-1}), \quad 
n=\overline{1,\infty}.
\end{eqnarray}
\end{lemma}
\begin{proof}  
Every $\sigma$-algebra  $ {\cal F}_n, \ n=\overline{1, \infty},$ has an exhaustive decomposition. The random value $f_n(\omega)$  satisfy all conditions of Theorems \ref{4vitanick3}, \ref{4vitanick4}.
This proves Lemma \ref{2myk10}. 
\end{proof}

\begin{thm}\label{2myk24}
  Let $M$ be  a  complete  set of    measures on the measurable space  $\{\Omega, {\cal F}\}$ with the filtration ${\cal F}_n$ on it.
 Then,  every non negative super-martingale  $\{f_n, {\cal F}_n\}_{n=0}^\infty $  relative to the set  of  measures $M$ 
 is a local regular one.
\end{thm}
\begin{proof}
 Without loss of generality, we assume that $f_n >a>0, n=\overline{1,N}.$ Then, the random value  $\frac{f_n}{f_{n-1}}$ satisfy the conditions of  Theorems   \ref{4vitanick4},  \ref{4vitanick7}. Therefore, all conditions of Theorem \ref{nick1} are satisfied. Theorem \ref{2myk24} is proved.
\end{proof}

\section{Local regularity of  majorized  super-martingales.}
In this section, we give the elementary proof that a majorized super-martingale relative to a complete set of measures  is a local regular one.

\begin{thm}\label{Tinmyk1}
On the measurable space $\{\Omega, {\cal F}\}$
 with the filtration ${\cal F}_n$ on it, let $M$ be a  complete set of  measures. Then, every bounded super-martingale $\{f_n, {\cal F}_n\}_{n=0}^\infty $ relative to  the  set of measures $M$  is a local regular one.
\end{thm}
\begin{proof}  From Theorem \ref{Tinmyk1} conditions, there exists a constant  $0<C<\infty$ such that  $|f_n| \leq C, \ n=\overline{1, \infty}.$  Consider  the super-martingale $\{f_n+C, {\cal F}_n\}_{n=0}^\infty .$ Then, $0\leq f_n+C\leq 2C.$  Due to Theorem  \ref{2myk24},
for the super-martingale $\{f_n+C, {\cal F}_n\}_{n=0}^\infty $ the local regularity is true. So, the same statement is valid  for the super-martingale $\{f_n, {\cal F}_n\}_{n=0}^\infty. $ Theorem \ref{Tinmyk1} is proved.
\end{proof}
 The next Theorem is  analogously proved as Theorem \ref{Tinmyk1}.
\begin{thm}\label{Tinmyk2}
On the measurable space $\{\Omega, {\cal F}\}$
 with the filtration ${\cal F}_n$ on it, let  $M$ be a  complete set of  measures. Then, a  super-martingale $\{f_n, {\cal F}_n\}_{n=0}^\infty $ relative to  the   set of  measures $M,$ satisfying the conditions 
\begin{eqnarray}
|f_n|\leq C_1 \xi_0, \quad f_n+C_1 \xi_0\leq C_2, \quad n=\overline{1,\infty}, \quad \xi_0 \in A_0,  
\end{eqnarray}
 for certain constants  $0<C_1, C_2<\infty,$  is  a local regular one.
\end{thm}

\section{Discrete geometric Brownian motion.}

In this section, we construct for the discrete evolution of risk assets the set of 
equivalent martingale measures and give a new formula for the fair price of super-hedge.
Let $\Omega_1^0=R^1, \  {\cal F}_1^0=B(R^1),$ where $R^1$ is a real axis, $B(R^1)$ is a Borel $\sigma$-algebra of $R^1.$ Let us put $\Omega_i=\Omega_1,  \  {\cal F}_i^0= {\cal F}_1^0, \  i=\overline{1, \infty},$ and
let us construct the infinite direct  product of the measurable spaces $\{\Omega_i^0, {\cal F}_i^0\}, \ i=\overline{1, \infty}. $ 
Denote $\Omega=\prod\limits_{i=1}^\infty\Omega_i^0.$ Under the $\sigma$-algebra ${\cal F}$ on $\Omega,$ we understand the minimal $\sigma$-algebra generated by sets $\prod\limits_{i=1}^\infty G_i, \  G_i \in {\cal F}_i^0,$ where in the last product  only the finite set of $G_i$ do not equal $\Omega_i^0.$
On the measurable space $\{\Omega, {\cal F}\},$ under the filtration  ${\cal F}_n$  we understand the minimal 
$\sigma$-algebra, generated by sets $\prod\limits_{i=1}^\infty G_i, \  G_i \in {\cal F}_i^0,$ where $G_i=\Omega_i^0$ for  $ i>n.$
Suppose that the points  $ t_0=0, t_1, t_2, \ldots, t_n, \ldots,$ belongs to $ R_1^+$ with $\Delta t=t_i - t_{i-1}$  not depending on the index $i.$
Let us  consider the  probability space
 $\{\Omega, {\cal F}, P\}, $ where $P=\prod\limits_{i=1}^\infty P_i^0, \ P_i^0=P_1^0, \ i=\overline{1, \infty},$
\begin{eqnarray}\label{5vitanick1}
P_1^0(A)= \frac{1}{[ 2\pi \Delta t]^{1/2}}\int\limits_{A}e^{-\frac{y^2}{2 \Delta t}}dy, \quad A  \in {\cal F}_1^0.
\end{eqnarray}
Define on the set $ t_0=0, t_1, t_2, \ldots, t_n, \ldots,$ the discrete Brownian motion. We say that the random process $w(t_i),\  i=\overline{0,\infty}, $ is a discrete    
Brownian motion, if on $\{\Omega, {\cal F}\}$ the joint distribution function is given by the formula
$$P_0(w(t_{i_1})\in A_{i_1}, \ldots, w(t_{i_k})\in A_{i_k} )= $$
\begin{eqnarray}\label{5vitanick2}
\frac{1}{D }\int\limits_{A_{i_1}\times \ldots \times A_{i_k} }e^{-\frac{y_{i_1}^2}{2 \Delta t_{i_1}}}\times \ldots \times e^{-\frac{[y_{i_k} - y_{i_{k-1}}]^2}{2 \Delta t_{i_k}} }dy_{i_1}\ldots dy_{i_k} , \quad A_{i_s}   \in {\cal F}_{i_s}^0,
\end{eqnarray}
$$ D=[ 2\pi]^{k/2} [ \Delta t_{i_1}\times \ldots \times \Delta t_{i_k}]^{1/2},  \  \Delta t_{i_s}=t_{i_s}- t_{i_{s-1}}.$$
So defined above the random process $w(t_i)$ on the set $ t_0, t_1, t_2, \ldots, t_n, \ldots,$ with $w(0)=0,$ is a homogeneous  one relative to the displacement on $k \Delta t, $ where $ \  k \geq 1,$ and is a natural number,  with the independent increments, the zero expectation and the correlation function 
$E^{P_0}w(t_s) w(t_k)=\min\{t_s, t_k\}.$

We assume that the evolution of non risk asset is given by the formula $B_n=e^{r t_n}, \ n=\overline{0,\infty}, $ where $r$ is an interest rate.
Let us consider on $\{\Omega, {\cal F}, P\}$ two cases of evolutions of risk assets given by the laws
\begin{eqnarray}\label{105vitanick3}
\bar S_n=S_0 e^{\sigma w(t_n)},
\end{eqnarray}
\begin{eqnarray}\label{105vitanick4}
\bar S_n=S_0 e^{(\mu -\frac{\sigma^2}{2})t_n+\sigma w(t_n)}.
\end{eqnarray}
Further, we consider the discount evolutions of the risk assets
\begin{eqnarray}\label{5vitanick3}
S_n=\frac{\bar S_n}{B_n}=S_0 e^{\sigma w(t_n)- r t_n},
\end{eqnarray}
\begin{eqnarray}\label{5vitanick4}
S_n=\frac{\bar S_n}{B_n}=S_0 e^{(\mu -\frac{\sigma^2  }{2} - r)t_n+\sigma w(t_n)}.
\end{eqnarray}

 It is convenient  to present these evolutions in the form
\begin{eqnarray}\label{5vitanick5}
S_n=(1+\rho_n) S_{n-1},  \quad  n=\overline{1,\infty},
\end{eqnarray}
with $\rho_n=e^{\sigma (w(t_n) -w(t_{n-1})) -r\Delta t}-1,$   $\rho_n=e^{(\mu -\frac{\sigma^2}{2}- r)\Delta t+\sigma (w(t_n) -w(t_{n-1}))}-1,$ correspondingly.

 On  the  probability space $\{\Omega, {\cal F}, P\} $  with the filtration  ${\cal F}_n$ on it, for further investigations it is convenient to present  the Brownian motion in equivalent form. We present the Brownian motion by the sequence of random values
$\zeta_n=\sum\limits_{i=1}^n y_i,  y_i \in \Omega_i^0, \ n=\overline{1, \infty},$
with the joint distribution functions 
$$P(\zeta_{i_1} \in A_{i_1}, \ldots, \zeta_{i_k} \in A_{i_k} )= $$
\begin{eqnarray}\label{6vitanick1}
\frac{1}{D }\int\limits_{A_{i_1}\times \ldots \times A_{i_k} }e^{-\frac{y_{i_1}^2}{2 i_1 \Delta t }} \times \ldots \times 
e^{-\frac{[y_{i_k} - y_{i_{k-1}}]^2}{2  (i_k - i_{k-1}) \Delta t}} dy_{i_1}\ldots dy_{i_k} , \quad A_{i_s}   \in {\cal F}_{i_s}^0,
\end{eqnarray}
$$ D=[ 2\pi]^{i_k/2} [ \Delta t] ^{i_k/2}[  i_1(i_2 - i_1)\times \ldots \times ({i_k}- {i_{k-1}})]^{1/2}.$$
Then, the discount  evolutions of the risk assets we can rewrite in the form 
\begin{eqnarray}\label{6vitanick2}
S_n=S_0 e^{\sigma \zeta_n -n r \Delta t},
\end{eqnarray}
\begin{eqnarray}\label{6vitanick3}
S_n=S_0 e^{(\mu -\frac{\sigma^2}{2} - r) n \Delta t+\sigma  \zeta_n}.
\end{eqnarray}
It is convenient  to present these discount evolutions in the form
\begin{eqnarray}\label{6vitanick4}
S_n=(1+\rho_n(y_n)) S_{n-1},  \quad  n=\overline{1,\infty},
\end{eqnarray}
with $\rho_n(y_n)=e^{\sigma y_n -r \Delta t} -1=\rho_1(y_n),$   $\rho_n(y_n)=e^{(\mu -\frac{\sigma^2}{2} - r)\Delta t+\sigma y_n}-1=\rho_1(y_n),$ correspondingly.

On the measurable space $\{\Omega^N, {\cal F}^N\}$ with the filtration $  {\cal F}_n, \ n=\overline{0,N},$ on it, where $\Omega^N=\prod\limits_{i=1}^N \Omega_i^0, \  {\cal F}^N=\prod\limits_{i=1}^N{\cal F}_i^0,$  let us introduce into consideration the set of measures $M^N.$ A measure  $Q$ belongs to $ M^N,$ if $Q=\prod\limits_{i=1}^N Q_i,$ where $ Q_i \in M_1^0 $ and for every $\bar Q \in M_1^0$  the representation
$$\bar Q(A)=\int\limits_{\Omega_1^-}\int\limits_{\Omega_1^+}\chi_{A}(y_1)\alpha(y_1, y_2)\frac{\rho_1^+(y_2)}{\rho_1^-(y_1)+\rho_1^+(y_2)}d\mu(y_1,y_2)+$$
\begin{eqnarray}\label{5vitanick6} \int\limits_{\Omega_1^-}\int\limits_{\Omega_1^+}\chi_{A}(y_2)\alpha(y_1, y_2)\frac{\rho_1^-(y_1)}{\rho_1^-(y_1)+\rho_1^+(y_2)}d\mu(y_1,y_2), \quad A \in {\cal F}_1^0,
\end{eqnarray}
$$ \Omega_1^-=\{y \in R^1, \rho_1(y) \leq 0\}=\{y \in R^1, y \leq \frac{r \Delta t}{\sigma}\}, $$ 
$$  \Omega_1^+=\{y \in R^1, \rho_1(y) >0\}=\{y \in R^1, y > \frac{r \Delta t}{\sigma}\}, $$
is valid,  where $\rho_1(y)=\rho_1^+(y) -\rho_1^-(y),$   $\rho_1(y)=e^{\sigma  y -r \Delta t}-1, $  $\mu=P^-\times P^+,$  
$$P^-(A)=\frac{1}{[2\pi \Delta t]^{1/2}}\int\limits_{A}e^{-\frac{y^2}{2\Delta t}}dy,  \quad A \in B(\Omega_1^-),$$
$$P^+(A)=\frac{1}{[2\pi \Delta t]^{1/2}}\int\limits_{A}e^{-\frac{y^2}{2\Delta t}}dy,  \quad A \in B(\Omega_1^+).$$
On the measurable space 
$\{  \Omega_1^-\times  \Omega_1^+, B(\Omega_1^-)\times  B(\Omega_1^+)\},$ the random value  $\alpha(y_1, y_2)$ 
satisfy the conditions:
\begin{eqnarray}\label{5vitanick7}
\mu(\{(y_1,y_2) \in \Omega_1^- \times \Omega_1^+, \    \alpha(y_1, y_2)>0\})=P(\Omega_1^+)P(\Omega_1^-), 
\end{eqnarray}
\begin{eqnarray}\label{5vitanick8}  \int\limits_{\Omega_1^-}\int\limits_{\Omega_1^+}\alpha(y_1, y_2)\frac{\rho_1^-(y_1)\rho_1^+(y_2)}{\rho_1^-(y_1)+\rho_1^+(y_2)}d\mu(y_1,y_2)<\infty,
\end{eqnarray}
\begin{eqnarray}\label{5vitanick9}
\int\limits_{\Omega_1^-}\int\limits_{\Omega_1^+}\alpha(y_1, y_2)d\mu(y_1,y_2)=1.
\end{eqnarray}
 Every   bounded random value  $ \alpha(y_1, y_2)>0, (y_1, y_2) \in R^-\times R^+,$ satisfy the conditions   (\ref{5vitanick7}) - (\ref{5vitanick9}), if $\sigma< \frac{1}{2\Delta t}, $ since $E^{P_1^0}|\rho_1(y)|< \infty.$
It means that the set of equivalent  martingale measures $M^N$ for the discount evolution $S_n=S_0 e^{\sigma \zeta_n -r n \Delta t}$ of the  risk asset   contains more then one martingale measure. In this case, the financial market is an  incomplete one. 

Denote $M_0^N=M^N_c$ the convex linear span  of the set of measures $M^N.$
On the measurable space $\{\Omega^N, {\cal F}^N\}$ with the filtration $  {\cal F}_n, \ n=\overline{0,N},$ on it, in correspondence with Theorem \ref{vitusjanick2}, the set of measures $M_0^N$ is a regular set of measures with the random variable $\xi_0=\prod\limits_{i=1}^N(1+\rho_i(y_i)),$
since the random value $ \eta_1 = \rho_1(y_1), $ figuring in Theorem \ref{vitusjanick2},  is an  integrable one relative to the measure $P_1^0$ and, therefore, $E^Q\xi_0=1, Q \in M_0^N.$

\begin{thm}\label{7vitanick1}
On the measurable space $\{\Omega^N, {\cal F}^N\}$ with the filtration $  {\cal F}_n, \ n=\overline{0,N},$ on it,  let the discount  risk asset evolution is given by the formula  $S_n=S_0 e^{\sigma \zeta_n- n r  \Delta t}$ for $\sigma  < \frac{1}{2 \Delta t}.$ For the payment function  $f(S_N),$ satisfying the condition $\sup\limits_{Q \in M_0^N} E^Q f(S_N)< \infty,$ the fair price of super-hedge is giving by the  formula 
$$\sup\limits_{Q \in M_0^N} E^Qf(S_N)=$$
$$\sup\limits_{\{y_i^1 \leq -d, \  y_i^2> - d, \  i=\overline{1,N}\}}\sum\limits_{i_1=1,\ldots, i_N=1}^2 f\left(S_0\prod\limits_{s=1}^N (1+\rho(y_s^{i_s}))\right)\times$$
\begin{eqnarray}\label{5vitanick10}
\prod\limits_{s=1}^N \frac{|e^{\sigma (d+y_s^{i_s+1})} -1|}{|e^{\sigma (d +y_s^{i_s+1})}-e^{\sigma(d+ y_s^{i_s})}|},
\end{eqnarray}
where we put $d=-\frac{r \Delta t}{\sigma},$ \ $ y_s^3=y_s^1.$
\end{thm}
\begin{proof}
The Borel $\sigma$-algebra $B(R^1)$ is generated by the exhaustive decomposition, since it has the countable set of intervals  with the rational number ends that generate $B(R^1).$ Therefore, the filtration  ${\cal F}_n, \ n=\overline{1, N}, $ has the exhaustive decomposition, due to Remark \ref{g1}. Theorem  \ref{4vitanick7} guarantee the formula for the fair price of super-hedge \cite{GoncharNick}. 
Due to Remark \ref{vitakolja1} after Theorem \ref{vitusjanick2}, the set of measures $\prod\limits_{i=1}^n \mu_{\{y_i^1, y_i^2\}},$ where
\begin{eqnarray}\label{7vitanick2}
\mu_{\{y_i^1, y_i^2\}}(A)=\chi_{A}(y_i^1)\frac{\rho_i^+( y_i^2)}{\rho_i^-( y_i^1)+\rho_i^+( y_i^2)}+\chi_{A}(y_i^2)\frac{\rho_i^-( y_i^1)}{\rho_i^-( y_i^1)+\rho_i^+( y_i^2)}, 
\end{eqnarray} 
$$ (y_i^1, y_i^2) \in \Omega_i^-\times \Omega_i^+, \ \Omega_i^-=\Omega_1^-, \ \Omega_i^+=\Omega_1^+, \  i=\overline{1,N},$$
forms the extreme points of the convex set of measures $M^N_0.$ 
The formula (\ref{5vitanick10}) is obtained by integration relative to the measure $\prod\limits_{i=1}^n \mu_{\{y_i^1, y_i^2\}}$ of the random value $f(S_N)$ and taking the $\sup$ on the set of all extreme points.
This prove the Theorem \ref{7vitanick1}.
\end{proof}

Now, let us consider the case, as  $\rho_n(y_n)=e^{(\mu -\frac{\sigma^2}{2}- r )\Delta t+\sigma y_n}-1=\rho_1(y_n).$

On the measurable space  $\{\Omega^N, {\cal F}^N\}$ with the filtration ${\cal F}_n, \ n=\overline{1, N}, $ on it,  where $ \Omega^N=\prod\limits_{i=1}^N \Omega_i^0, \  {\cal F}^N=\prod\limits_{i=1}^N{\cal F}_i^0,$  we  introduce into consideration the set of measures $M^N.$  A measure  $Q$ belongs to $ M^N,$ if $Q=\prod\limits_{i=1}^N Q_i, \  Q_i \in M_1^0.$ For every  $\bar Q \in M_1^0$  the representation
$$\bar Q(A)=\int\limits_{\Omega_1^-}\int\limits_{\Omega_1^+}\chi_{A}(y_1)\alpha(y_1, y_2)\frac{\rho_1^+(y_2)}{\rho_1^-(y_1)+\rho_1^+(y_2)}d\mu(y_1,y_2)+$$
\begin{eqnarray}\label{11vitanick6} \int\limits_{\Omega_1^-}\int\limits_{\Omega_1^+}\chi_{A}(y_2)\alpha(y_1, y_2)\frac{\rho_1^-(y_1)}{\rho_1^-(y_1)+\rho_1^+(y_2)}d\mu(y_1,y_2), \quad A \in {\cal F}_1^0,
\end{eqnarray}
$$ \Omega_1^-=\{y \in R^1, \rho_1(y) \leq 0\}=\left\{y \in R^1, y \leq - \frac{(\mu -\frac{\sigma^2}{2}- r)\Delta t}{\sigma}\right\},$$ $$  \Omega_1^+=\{y \in R^1, \rho_1(y) >0\}=\left\{y \in R^1, y >- \frac{(\mu -\frac{\sigma^2}{2} - r)\Delta t}{\sigma}\right\}, $$
is valid,  where $\rho_1(y)=\rho_1^+(y) -\rho_1^-(y),$   $\rho_1(y)=e^{(\mu -\frac{\sigma^2 }{2} - r)\Delta t+\sigma y}-1, $  $\mu=P^-\times P^+,$  
$$P^-(A)=\frac{1}{[2\pi \Delta t]^{1/2}}\int\limits_{A}e^{-\frac{y^2}{2\Delta t}}dy,  \quad A \in B(\Omega_1^-),$$
$$P^+(A)=\frac{1}{[2\pi \Delta t]^{1/2}}\int\limits_{A}e^{-\frac{y^2}{2\Delta t}}dy,  \quad A \in B(\Omega_1^+).$$
On the measurable space 
$\{  \Omega_1^-\times  \Omega_1^+, B(\Omega_1^-)\times  B(\Omega_1^+)\},$ the random value  $\alpha(y_1, y_2)$ 
satisfy the conditions
\begin{eqnarray}\label{11vitanick7}
\mu(\{(y_1,y_2) \in \Omega_1^- \times \Omega_1^+, \    \alpha(y_1, y_2)>0\})=P(\Omega_1^+)P(\Omega_1^-), 
\end{eqnarray}
\begin{eqnarray}\label{11vitanick8}  \int\limits_{\Omega_1^-}\int\limits_{\Omega_1^+}\alpha(y_1, y_2)\frac{\rho_1^-(y_1)\rho_1^+(y_2)}{\rho_1^-(y_1)+\rho_1^+(y_2)}d\mu(y_1,y_2)<\infty,
\end{eqnarray}
\begin{eqnarray}\label{11vitanick9}
\int\limits_{\Omega_1^-}\int\limits_{\Omega_1^+}\alpha(y_1, y_2)d\mu(y_1,y_2)=1,
\end{eqnarray}
for every   bounded $ \alpha(y_1, y_2)>0,$  if $\sigma< \frac{1}{2\Delta t},$  since $E^{P_1^0}|\rho_1(y)|< \infty.$
Denote $M_0^N=M^N_c$ the convex linear span  of the set of measures $M^N.$
On the measurable space $\{\Omega^N, {\cal F}^N\}$ with the filtration $  {\cal F}_n, \ n=\overline{0,N},$ on it, in correspondence with  Theorem \ref{vitusjanick2}, the set of measures $M_0^N$ is a regular set of measures with the random variable $\xi_0=\prod\limits_{i=1}^N(1+\rho_i(y_i)),$
since the random value $ \eta_1 = \rho_1(y_1), $ figuring in  Theorem \ref{vitusjanick2},  is an  integrable one relative to the measure $P_1^0$ and, therefore, $E^Q\xi_0=1, Q \in M_0^N.$
It means that the set of equivalent  martingale measures $M^N_0$ for the discount  evolution $S_n=S_0 e^{(\mu - \frac{\sigma^2}{2} - r) n \Delta t+ \sigma \zeta_n}$ of the risk asset contains more then one martingale measure. In this case, the financial market is an incomplete one. 

\begin{thm}\label{8vitanick1}
On the measurable space $\{\Omega^N, {\cal F}^N\}$ with the filtration $  {\cal F}_n, \ n=\overline{0,N},$ on it, let the discount risk asset evolution is given by the formula  $S_n=S_0 e^{(\mu - \frac{\sigma^2}{2} -r) n \Delta t+ \sigma \zeta_n}$  for $\sigma  < \frac{1}{2 \Delta t}.$ For the payment function  $f(S_N),$ satisfying the condition $\sup\limits_{Q \in M_0^N} E^Qf(S_N)< \infty, $ the fair price of super-hedge is giving by the  formula 
$$\sup\limits_{Q \in M_0^N} E^Qf(S_N)=$$
$$\sup\limits_{y_i^1 \leq - d, \ y_i^2>- d, \ i=\overline{1,N}}\sum\limits_{i_1=1,\ldots,i_N=1}^2 f\left(S_0\prod\limits_{s=1}^N (1+\rho(y_s^{i_s}))\right)\times$$
\begin{eqnarray}\label{11vitanick10}
\prod\limits_{s=1}^N \frac{|e^{\sigma (d + y_s^{i_s+1})} -1|}{|e^{\sigma (d + y_s^{i_s+1})}-e^{\sigma (d +  y_s^{i_s})}|},
\end{eqnarray}
 where we put $d=\frac{(\mu - \frac{\sigma^2}{2} -r )\Delta t}{\sigma},$ $ y_s^3=y_s^1.$
\end{thm}
The proof of Theorem \ref{8vitanick1} is the same as the proof of   Theorem \ref{7vitanick1}.

\section{Conclusions.}

In the paper, we  generalize the results of the paper \cite{GoncharNick}.
Section 2 contains the definition of local regular super-martingales. Theorem  \ref{reww1} gives the necessary and sufficient conditions of the local regularity of  a super-martingale. In spite of its simplicity,  the Theorem \ref{reww1} appeared very useful for the description of the local regular super-martingales.

Section 3 contains the important  Definition \ref{mykvita1} of the set of equivalent measures consistent with the filtration. In Lemma \ref{tinnick1}, we give an example of the set of equivalent measures
consistent with the filtration. Theorem \ref{tatnick2} contains  the sufficient conditions under that there exists a nonnegative super-martingale on a measurable space with  the set of measures consistent with the filtration. In Theorem \ref{tatnick8}, the sufficient conditions are founded which guarantee the existence on a measurable space a regular martingale.
 
Lemma \ref{vitanick1} gives the sufficient conditions of the  existence of a  set of  measures consistent with the filtration.

Lemma \ref{vitanick9} contains the description of the set of measures being equivalent to a given measure  and satisfying the condition: mathematical expectation of a given random value relative to every such a measure equals zero.
 In Lemma \ref{101vitanick13}, we obtain the representation for the set of measures being equivalent to a given measure and satisfying the condition: the conditional expectation of a given random value relative to every of which equals zero.
At last, Theorem \ref{nick1} gives the necessary and sufficient conditions of the local regularity of a nonnegative super-martingale.

In Section 4, 
 in Lemma \ref{vitanick28}, we investigate the closure of the set of considered set of measure in  the case of the countable space of elementary events. It is proved that in metrics (\ref{vitanick27}) the closure of the set of considered set of measures contains the set of measures (\ref{vitanick29}). 

Further, we introduce the notion of the exhaustive decomposition   of a measurable space. Using this notion, in Lemma \ref{2vitanick1}, we describe the closure of the considered  set of measures relative to the pointwise convergence of measures and the closure of expectation values relative to this set of measures. 

Theorem \ref{koljavita1} is a consequence of Lemma \ref{vitanick9} and contains the description of the set of measures, being equivalent to the given measure,  expectations relative to which are equal one. Theorem \ref{koljavita7} states the necessary and sufficient conditions when the set of measures (\ref{koljavita3}) is consistent with filtration. 
In Theorem \ref{mykolvit1}, we give the necessary and sufficient conditions of the consistency with the filtration of the set of measures  (\ref{koljavita3}).

Theorem \ref{mykolvit1} states the necessary and sufficient conditions
of the consistency with the filtration of the set of measure (\ref{koljavita3}).
 Using Lemma \ref{vitanick9}, in Lemma \ref{1vitanick35},  we construct an example of the  set of  measures consistent with the filtration. 
In Theorem \ref{vitusjanick2}, we describe completely the local regular set of measures. 

In Definition \ref{211vitanick}, we introduce a fundamental notion of the completeness of the regular set of  measures.

Using Lemma \ref{vitanick28} and \ref{2vitanick1},  Theorem \ref{4vitanick3} states that the expectations of the integrable random values  relative to the contraction of the complete  set of measures on the  $\sigma$-algebras of filtration contains the points (\ref{4vitanick2}).

Theorem \ref{4vitanick4} states that for every nonnegative  ${\cal F}_n$ measurable random value, mathematical expectation for which relative to every martingale measure is bounded by 1, the inequality (\ref{4vitanick5}) is true.

In Theorem \ref{4vitanick7},  it is proved that every nonnegative super-martingale relative to the  regular set of measures is a local regular one.
 The same statement, as in Theorem \ref{4vitanick7}, it is proved in Theorem \ref{400vitanick7}  in the case, as a super-martingale is bounded from below.

Section 5 contains the description of the local regular super-martingales.
Using Theorem  \ref{reww1}, we prove Theorem \ref{mmars1}, giving the possibility
to describe  the local regular super-martingales.
 Further, we introduce a class $K$ of the local regular super-martingales relative to a regular set of measures. 
Theorem \ref{mmars9} states that every nonnegative  uniformly integrable
super-martingale relative to a regular set of measures belong to the class $K.$
The next Theorem \ref{9mmars9} states that all super-martingales that are majorized by elements from the set $A_0$ is also belong to the class $K.$
At last, in corollary \ref{mars16}, we give an example of the local regular super-martingele playing important role in the definition of the fair price of the contingent claim \cite{GoncharNick}. 

Section 6 contains an application of the results obtained above.
To make this helps us Theorem \ref{nick1} giving the necessary and sufficient conditions of  the local regularity of the nonnegative super-martingales.
In subsection 6.1, we consider the applications of the results obtained in the case as $\sigma$-algebra on the set of elementary events is generated by the finite set of events. In this case, Lemma \ref{1myk10} states that inequality   (\ref{1myk11})
is true. Theorem \ref{1myk19} states that every nonnegative super-martingale is local regular one. The same statement is true, when a super-martingale is only bounded, as it is shown in Theorem \ref{myktinal4}.
In subsection 6.2, we consider the measurable space with the countable decomposition. In Lemma \ref{2myk10}, we obtain the inequality (\ref{2myk11}).
Theorem \ref{2myk24} states that every nonnegative super-martingale is a local regular one.

Section 7 contains two statements.The first  statement is  that every bounded super-martingale is a local regular one. It is contained in Theorem \ref{Tinmyk1}.   
The second statement is contained in Theorem \ref{Tinmyk2}.  It declares that a majorized super-martingale is also a local regular one.

Section 8 contains the application of the results obtained above to calculation of the fair price of super-hedge, when the risk asset evolves by the discrete geometric Brownian motion.
In this case, we  describe the set of regular measures. We find the set of extreme points of the regular set of measures.  It is proved that the  the fair price of the super-hedge is given by the formula (\ref{11vitanick10}).

\vskip 5mm

\end{document}